\documentclass[12pt, conference, onecolumn]{IEEEtran}

\usepackage{url}
\usepackage{ifthen}
\usepackage{cite}
\usepackage[cmex10]{amsmath} % Use the [cmex10] option to ensure complicance
\usepackage{balance}
\usepackage{amsmath}
\usepackage{amsthm}
\usepackage{amssymb}
\usepackage{multirow}
\usepackage{amsfonts}
\usepackage{graphicx}
\usepackage{algorithm}
\usepackage{url}
\usepackage{lipsum}
\usepackage{booktabs,makecell,multirow}
\usepackage{hyperref}

\usepackage{latexsym}
\usepackage{color}
\usepackage{graphicx}
\usepackage{threeparttable}
\usepackage{float}

\usepackage{algpseudocode}% http://ctan.org/pkg/algorithmicx
\usepackage{subcaption}

\errorcontextlines\maxdimen

\DeclareGraphicsExtensions{.tif,.pdf,.jpeg,.png,.jpg,.bmp,.svg,.eps,.EMF}

\newtheorem{theorem}{Theorem}%[section]
\newtheorem{definition}{Definition}%[section]
\newcommand{\cM}{\mathcal{M}}

\newcommand{\cT}{\mathcal{T}}

\newcommand{\cG}{\mathcal{G}}

\newcommand{\bX}{\mathbf{x}}
\newcommand{\bY}{\mathbf{y}}

\newcommand{\bV}{\mathbf{v}}

\newcommand{\sfH}{\mathsf{H}}
\newcommand{\sfF}{\mathsf{F}}
\newcommand{\sfX}{\mathsf{X}}

\newcommand{\poly}{\mathsf{poly}}

\newcommand{\supp}{\mathsf{supp}}

\begin{document}

\title{Improved non-adaptive algorithms for threshold group testing with a gap\footnote{This paper has appeared in part at the 2020 IEEE International Symposium on Information Theory.}\\[.7ex] 
  {\normalfont\large 
	Thach V. Bui\IEEEauthorrefmark{1}\IEEEauthorrefmark{2}, Mahdi Cheraghchi\IEEEauthorrefmark{3}, and Isao Echizen\IEEEauthorrefmark{4}\IEEEauthorrefmark{5}}\\[-1.5ex]}

%\begingroup
%\centering
%\endgroup
%\author{\IEEEauthorblockA{\IEEEauthorrefmark{1}SOKENDAI (The \\Graduate University \\for Advanced \\Studies), Hayama, \\Kanagawa, Japan\\ bvthach@nii.ac.jp}
%\and
%\IEEEauthorblockA{\IEEEauthorrefmark{3}Graduate School\\ of Natural Science\\ and Technology, \\Okayama University, \\Okayama, Japan\\kminoru@okayama-u.ac.jp}
%\and
%\IEEEauthorblockA{\IEEEauthorrefmark{4}Department of \\Computing, Imperial \\College London, UK\\m.cheraghchi@imperial.ac.uk}
%\and
%\IEEEauthorblockA{\IEEEauthorrefmark{2}National Institute\\ of Informatics, \\Tokyo, Japan \\ iechizen@nii.ac.jp}}
%

\author{\IEEEauthorblockA{\IEEEauthorrefmark{1}University of Science, VNU-HCMC, \\Ho Chi Minh City, Vietnam\\
\IEEEauthorrefmark{2}University of Padova, Padova, Italy\\ bvthach@fit.hcmus.edu.vn}
\and
\IEEEauthorblockA{\IEEEauthorrefmark{3}University of Michigan, \\Ann Arbor MI, USA\\mahdich@umich.edu}
\and
\IEEEauthorblockA{\IEEEauthorrefmark{4}National Institute of Informatics,\\ \IEEEauthorrefmark{5}The University of Tokyo,\\ Tokyo, Japan\\iechizen@nii.ac.jp}
}

\maketitle

\thispagestyle{plain}
\pagestyle{plain}

\IEEEpeerreviewmaketitle

\begin{abstract}
The basic goal of threshold group testing is to identify up to $d$ defective items among a population of $n$ items, where $d$ is usually much smaller than $n$. The outcome of a test on a subset of items is positive if the subset has at least $u$ defective items, negative if it has up to $\ell$ defective items, where $0 \leq \ell < u$, and arbitrary otherwise. This is called threshold group testing. The parameter $g = u - \ell - 1$ is called \textit{the gap}. In this paper, we focus on the case $g > 0$, i.e., threshold group testing with a gap. Note that the results presented here are also applicable to the case $g = 0$; however, the results are not as efficient as those in related work. Currently, a few reported studies have investigated test designs and decoding algorithms for identifying defective items. Most of the previous studies have not been feasible because there are numerous constraints on their problem settings or the decoding complexities of their proposed schemes are relatively large. Therefore, it is compulsory to reduce the number of tests as well as the decoding complexity, i.e., the time for identifying the defective items, for achieving practical schemes.

The work presented here makes five contributions. The first is a more accurate theorem for a non-adaptive algorithm for threshold group testing proposed by Chen and Fu. The second is an improvement in the construction of disjunct matrices, which are the main tools for tackling (threshold) group testing and other tasks such as constructing cover-free families or learning hidden graphs. Specifically, we present a better exact upper bound on the number of tests for disjunct matrices compared with that in related work. The third and fourth contributions are a reduced exact upper bound on the number of tests and a reduced asymptotic bound on the decoding time for identifying defective items in a noisy setting on test outcomes. The fifth contribution is a simulation on the number of tests of the resulting improvements for previous work and the proposed theorems.
\end{abstract}


\section{Introduction}
\label{sec:intro}

Identification of up to $d$ defective items in a large population of $n$ items is the main objective of group testing. Defective items satisfy a specific property while negative (non-defective) items do not. Dorfman~\cite{dorfman1943detection}, an economist who served during World War II, initiated this research direction in an effort to identify syphilitic draftees among a large population of draftees. Rather than testing the draftees one by one, which would have taken much time and money, he proposed pooling the draftees into groups for testing, which is more efficient. Ideally, if there was at least one syphilitic draftee present in the group, the test outcome would be positive. Otherwise, it would be negative. This approach can be generalized by replacing ``draftee'' with ``item,'' ``syphilis'' with ``a specific property,'' and ``syphilitic draftee'' with ``defective item.'' This is classical group testing (CGT) without noise. Formally, in CGT without noise, the outcome of a test on a subset of items is positive if the subset has at least one defective item and negative otherwise. If noise is present, the outcome may flip from positive to negative and vice versa.

A generalization of CGT called \textit{threshold group testing} (TGT) was introduced by revising the definition of the test outcome~\cite{damaschke2006threshold}. In this model, the outcome of a test on a subset of items is positive if the subset has at least $u$ defective items, negative if it has up to $\ell$ defective items, where $0 \leq \ell < u$, and arbitrary otherwise. This model is denoted as $(n, d, \ell, u)$-TGT. The parameter $g = u - \ell - 1$ is called the gap. When $g = 0$, i.e., $\ell = u - 1$, threshold group testing has no gap. When $u = 1$, TGT reduces to CGT. TGT can be considered as a special case of complex group testing~\cite{chen2008upper} or generalized group testing with inhibitors~\cite{bui2018framework}. Like previous reports such as~\cite{damaschke2006threshold,cheraghchi2013improved,de2020subquadratic,chan2013stochastic,chen2009nonadaptive}, the focus of this paper is on threshold group testing with a gap, i.e., $g > 0$. Note that the results here are also applicable to the no-gap case ($g = 0$). However, this case should be treated separately to attain efficient solutions as presented in~\cite{d2013superimposed,de2020subquadratic,bui2019efficiently}.

In general, TGT is more complicated than CGT even for trivial testing since instead of testing all individuals as in CGT, all groups of a certain size (depending on the threshold parameters) have to be tested. It is intuitively obvious that the outcome of a test on a certain subset of items in TGT has less information than one in CGT. For example, if the outcome of a test on a subset of items is negative, we can be sure that there are no defectives in the subset if the test was done under the CGT setting, whereas the subset has up to $u - 1$ defectives if the test was done under the TGT setting.

We illustrate TGT for two thresholds ($u = 10$ and $\ell = 2$) versus CGT in Fig.~\ref{fig:ThresholdWGap}. The black and red dots represent negatives and defectives, respectively. A subset containing defectives and/or negatives is a blue circle containing black and/or red dots. The outcome of a test on a subset of items is positive (\textcolor{red}{$+$}) or negative ({$-$}). In CGT (``Classical'' in the figure), the outcome of a test on a subset of items is positive if the subset has at least one red dot, and negative otherwise. In TGT with two thresholds $u$ and $\ell$ (``Threshold'' in the figure), the outcome of a test on a subset of items is positive if the subset has at least $u = 10$ red dots, negative if the subset has up to $\ell = 2$ red dots, and arbitrary otherwise.

\begin{figure*}
\centering
\includegraphics[scale=0.5]{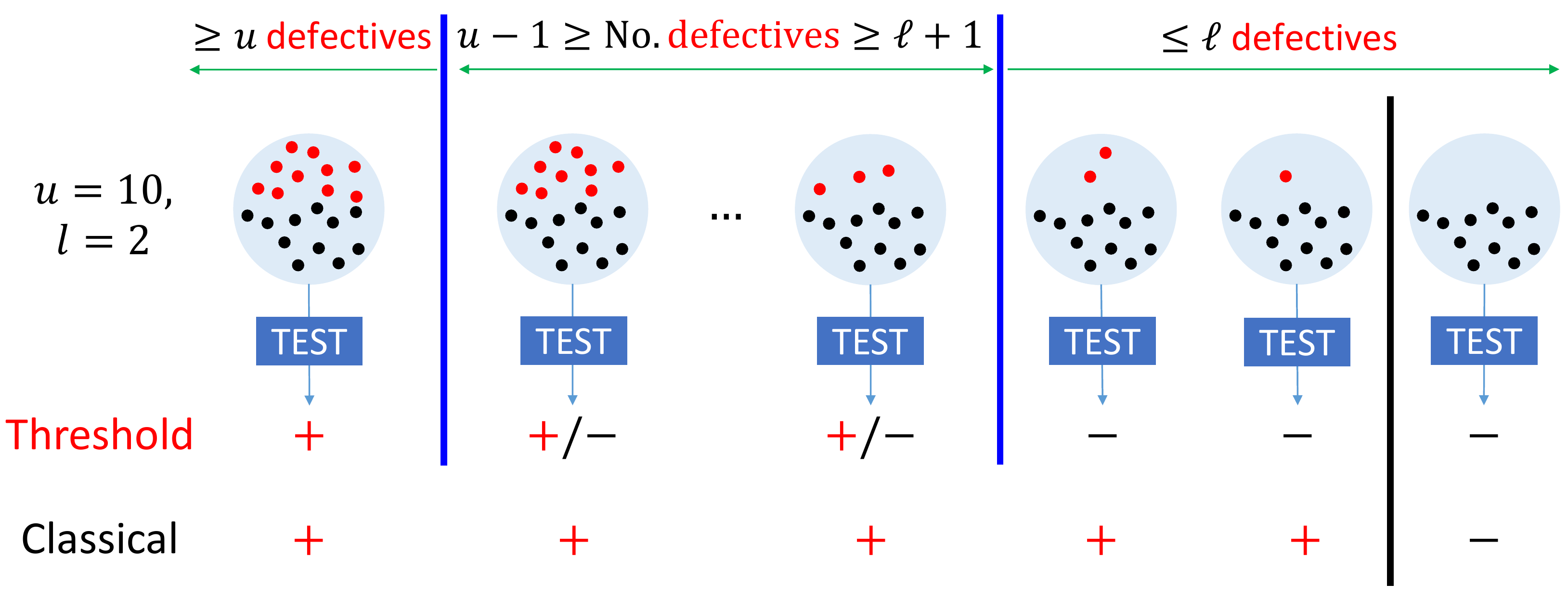}
\caption{Illustration of threshold group testing for $u = 10$ and $\ell = 2$ versus classical group testing.}
\label{fig:ThresholdWGap}
\end{figure*}

There are two approaches to designing tests. The first is \textit{adaptive group testing} (AGT) in which the design of a test depends on the designs of the previous tests. This approach usually achieves optimal bounds on the number of tests; however, it takes much time. The second is \textit{non-adaptive group testing} (NAGT) which is an alternative solution for AGT. With this approach, all tests are designed independently and can be performed in parallel. Because of the resulting time saving, NAGT has been widely applied in various fields such as computational and molecular biology~\cite{du2000combinatorial}, networking~\cite{d2019separable}, and neuroscience~\cite{bui2018framework}. Recently, group testing seems to be an efficient way to economically and quickly identify infected persons during the coronavirus pandemic of 2020--2021~\cite{shental2020efficient,gabrys2020ac}.

NAGT can be represented by a (binary) measurement matrix in which each row and each column represent a test and an item, respectively. An entry in the matrix at row $i$ and column $j$ that equals $1$ naturally means that item $j$ belongs to test $i$; and an entry that equals $0$ means otherwise. For every group testing problem, we have a few possible cases:
\begin{enumerate}
\item The ``for all'' model (the worst case): we have a single measurement matrix and the same matrix has to recover any set of up to $d$ defectives. Note that the matrix can be randomized or explicit, but once we have the matrix, the same matrix has to work correctly for any configuration of up to $d$ defectives.
\item The ``for each'' model: for every fixed set of defectives, when we sample a measurement matrix, with high probability (whp) from the measurement outcomes, we can reconstruct the defectives.
\item The ``average case'' model: the matrix can be explicit or random, and the defectives are chosen randomly (often uniform, or iid). Then, whp over all randomness involved, we should be able to recover the defectives from the measurements.
\end{enumerate}

NAGT generally refers to the ``for all'' model; otherwise, the model is specified. The focus of the work reported here is on NATGT (Non-Adaptive Threshold Group Testing), which is TGT associated with NAGT (with the ``for all'' model).

There are two main requirements for efficiently tackling group testing: minimize the number of tests and efficiently identify the set of defective items. Lengthy and intensive study of CGT has shown that the number of tests needed for effective use of AGT is $\Omega(d\ln{n})$~\cite{du2000combinatorial}, which is theoretically optimal. The decoding algorithm is usually included in the test design. For NAGT, Porat and Rothschild~\cite{porat2011explicit} first proposed explicit non-adaptive constructions using $O(d^2 \ln{n})$ tests with no efficient (sublinear to $n$) decoding algorithm. To have an efficient decoding algorithm, says $\poly(d, \ln{n})$, while keeping the number of tests as small as possible, says $O(d^{1 + o(1)} \ln^{1 + o(1)}{n})$, several schemes have been proposed~\cite{indyk2010efficiently,ngo2011efficiently,cheraghchi2013noise,bui2019efficient}. Using probabilistic methods, Cai et al.~\cite{cai2017grotesque} required only $O(d \ln{d} \cdot \ln{n})$ tests to find defective items in time $O(d(\ln{n} + \ln^2{d}))$. Recently, Bondorf et al.~\cite{bondorf2020sublinear} presented a bit mixing coding that achieves asymptotically vanishing error probability with $O(d \log{n})$ tests to identify defective items in time $O(d^2 \log{d} \cdot \log{n})$ as $n \rightarrow \infty$. For further reading, we recommend readers to refer to the survey in~\cite{aldridge2019group}.

From the genesis of TGT, Damaschke~\cite{damaschke2006threshold} showed that the set of defective items can be identified with up to $g$ false positives (i.e., negative items are identified as defective items) and $g$ false negatives (i.e., defective items are identified as negative items) by using $\binom{n}{u}$ non-adaptive tests. Chen et al.~\cite{chen2008upper} gave an upper bound on the number of tests: $t(n, d, u; z] = O \left( z \left( \frac{d + u}{u} \right)^u \left( \frac{d + u}{d} \right)^d (d + u) \ln{\frac{n}{d + u}} \right)$, where $\lfloor (z - 1)/2 \rfloor$ is usually referred to as the maximum number of errors in the test outcomes. Cheraghchi~\cite{cheraghchi2013improved} asserted that this bound is not optimal. Therefore, he reduced it to $O(d^{g+2} \ln(n/d) \cdot (8u)^u) = O(d^{g+2} \ln(n/d))$ tests under the assumption that $u$ is constant, which is asymptotically optimal. When $d = \ell + u$, Ahlswede et al.~\cite{ahlswede2011bounds} gave an upper bound on the number of tests, which is $O(u 2^{2u} \log{n})$. They also considered the case $d \neq \ell + u$; however, the bound on the number of tests has no constructive approximations for inference.

There have been a few studies on decoding algorithms for NATGT with a gap and with the ``for all'' or ``for each'' model. By using models for the gap and considering the ``for each'' model, Chan et al.~\cite{chan2013stochastic} set that the number of defective items to exactly $d$, $u = o(d)$, and used $O\left( \ln{\frac{1}{\epsilon}} \cdot d\sqrt{u} \ln{n}\right)$ tests to identify the defective items in time $O(n \ln{n} + n \ln{\frac{1}{\epsilon}})$, which is linear to the number of items, where $\epsilon \in (0, 1)$. Recently, by setting $d = O(n^\beta)$ for $\beta \in (0, 1)$ and $u = o(d)$, Reisizadeh et al. \cite{reisizadeh2018sub} use $\Theta(\sqrt{u} d \ln^3{n} )$ tests to identify all defective items in time $O(u^{1.5} d \ln^4{n} )$ whp with the aid of a $O(u \ln{n}) \times \binom{n}{u}$ look-up matrix, which is unfeasible when $n$ or $u$ is large. To the best of our knowledge, the first and only work to tackle the ``for all'' model in NATGT with a decoding algorithm is that by Chen and Fu~\cite{chen2009nonadaptive}. They proposed schemes for finding the defective items using $t(n, d - \ell, u; z]$ tests in time $O(n^u \ln{n})$. However, the decoding time becomes impractical as $n$ or $u$ increases.

We consider here the potential use of threshold group testing as a tool to tackle the problems in designing tests for detecting viral infections~\cite{dorfman1943detection,emad2014semiquantitative,gabrys2020ac} and chemical screening~\cite{damaschke2006threshold}. Damaschke~\cite{damaschke2006threshold} introduced threshold group testing with some potential applications for chemical screening, without presenting a concrete application. Back to the work of Dorfman~\cite{dorfman1943detection}, even for standard blood tests, the uncertainty in deciding the outcome of a test on a pool of blood samples remains problematic in practice. As mentioned in the first paragraph of this section, the outcome of a test on a pool of blood samples is positive if the pool contains at least one syphilitic sample and negative otherwise. However, in practice, before deciding the outcome of a test, we must get \textit{a reference value} associated with the test. A next procedure is to set a \text{threshold} such that the outcome of a test is positive if its reference value is larger than or equal to the threshold and negative otherwise. Due to the presence of impurities in blood sample pools, it is difficult to set a unique threshold. Dorfman suggested setting it to be the average of the impurities in the separate samples. This would result in three ranges for the threshold: \textit{positive}, \textit{negative}, and \textit{inconclusive}. If the threshold is in the positive (negative) range, the outcome of a test is positive (negative) if its reference value is larger (smaller) than or equal to the threshold. If the threshold is in the inconclusive range, it is uncertain to decide whether the test outcome is positive or negative. This is exactly what TGT with a gap tries to capture. In 2014, Emad and Milenkovic~\cite{emad2014semiquantitative} introduced ``semi-quantitative group testing'' (SQGT) to tackle a model for quantitative polymerase chain reaction (qPCR) tests. Since TGT is a special case of SQGT, it can also be used in qPCR tests. The work of Gabrys et al.~\cite{gabrys2020ac} motivated the application of TGT to reverse transcription PCR (RT-PCR) or quantitative PCR (qPCR) tests for viral infections such as Covid-19. The fluorescence values captured by the PCR process has different levels, and again one can assign a positive range, a negative range, and an inconclusive range in a manner similar to the work of Dorfman.

\subsection{Contributions}
\label{sub:intro:contri}

The focus of this work is TGT with a gap; i.e., $g = u - \ell - 1 > 0$. Note that the results here are also applicable to the no-gap case, i.e., $g = 0$; however, the no-gap case should be treated separately to attain efficient solutions, as explained in~\cite{d2013superimposed,de2020subquadratic,bui2019efficiently}.

The first contribution, which is summarized in Theorem~\ref{thr:corrected}, is correction of the decoding complexity analysis by Chen and Fu~\cite{chen2009nonadaptive}. Their inaccurate analysis in decoding complexity resulted in much smaller decoding complexity than the actual one.

The second contribution is a better exact upper bound on the number of tests of $(n, d, u; z]$-disjunct matrices (defined later). We significantly reduce the upper bound on the number of tests for constructing disjunct matrices compared with the work of Chen et al.~\cite{chen2008upper}. The basic idea is that instead of using a hypergraph to generate a disjunct matrix as Chen et al. did, we directly generate a random disjunct matrix. This improvement paves the way to improved results not only in group testing, but also in other fields such as graph learning~\cite{abasi2018non} and cover-free family construction~\cite{stinson2004generalized}.

The third and fourth contributions are a reduced exact upper bound on the number of tests and a reduced asymptotic bound on the decoding time for identifying defective items in a noisy setting on test outcomes compared with the state-of-the-art work of Chen and Fu~\cite{chen2009nonadaptive}. The number of tests is directly reduced by using a better upper bound on the number of tests (the second contribution). The basic idea for reducing decoding time is to pick subsets of potential defectives such that each subset contains at least $\ell + 1$ defectives and then return the union of these subsets as an approximate defective set. To attain a better approximate defective set (at the cost of a longer decoding time), the approximate defective set derived as described above is taken as the input to the existing algorithm in~\cite{chen2009nonadaptive}.

Suppose there are up to $\lfloor (z - 1)/2 \rfloor$ erroneous outcomes. Let $S^\prime$ be the approximate defective set returned by decoding procedure. Two sets $S \setminus S^\prime$ and $S^\prime \setminus S$ are referred to as the sets of false negatives and false positives, respectively. Chen and Fu~\cite{chen2009nonadaptive} use $t(n, d - \ell, u; z] = O \left( z \left( \frac{k}{u} \right)^u \left( \frac{k}{d - \ell} \right)^{d - \ell} k \ln{\frac{n}{k}} \right)$ tests to recover a set $S^\prime$ with $|S^\prime \setminus S| \leq g$ and $|S \setminus S^\prime| \leq g$, where $k = d - \ell + u$. By using $h(n, d - \ell, u; z] = O \left( \left( 1 + \frac{z}{\alpha} \right) \cdot \left( \frac{k}{u} \right)^u \left( \frac{k}{d - \ell} \right)^{d - \ell} k \ln{\frac{n}{k}} \right)$ tests where $k = d - \ell + u$ and $\alpha = k \ln{\frac{\mathrm{e} n}{k}} + u \ln{\frac{\mathrm{e} k}{u}}$, we can recover a set $S^\prime$ close to the true defective set $S$ as follows:

\begin{enumerate}
\item $|S^\prime \setminus S| \leq g$ and $|S \setminus S^\prime| \leq g$.
\item $|S^\prime \setminus S| \leq g w$ and $|S \setminus S^\prime| \leq g$, where $w = \left( \left\lfloor \frac{|S|}{\ell + 1} \right\rfloor + u - 1 \right)g$.
\item $|S^\prime \setminus S| \leq g$ and $|S \setminus S^\prime| \leq 2g$.
\end{enumerate}

The decoding complexities of these three cases are always smaller than the one (after correction) proposed by Chen and Fu~\cite{chen2009nonadaptive}. 

The last contribution is a simulation for previous work and our proposed theorems. The results demonstrate the superiority of our proposed theorems over previous ones and validate the arguments presented here.

The contributions are summarized in Theorems~\ref{thr:corrected},~\ref{thr:improvedChenUpper},~\ref{thr:improved0},~\ref{thr:mainImprovement}, and~\ref{thr:mainImprovement2} and illustrated in Fig.~\ref{fig:Overview} (except Theorem~\ref{thr:corrected}). The ovals, lines, parallelograms, and rectangles represent start or end point, connectors showing relationships between the representative shapes, inputs or outputs, and processes, respectively. The dash-dot line represents a comment on the representative shapes. The blue arrows represent the previous schemes while the other arrows represent our proposed theorems.

\begin{figure*}[t!]
\centering
\includegraphics[scale=0.5]{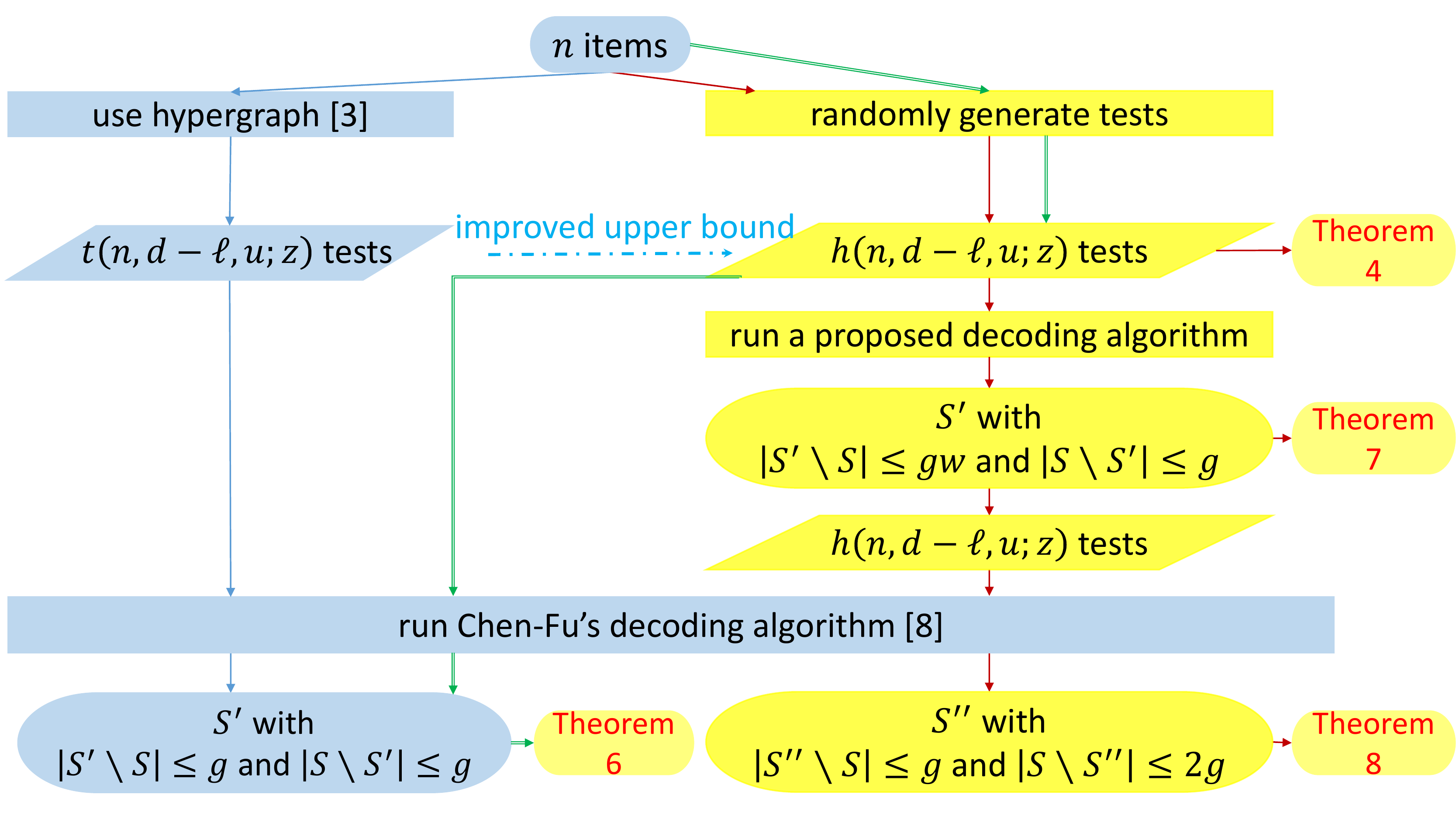}

\caption{Flow chart illustrating how contributions were attained in this work (excluding Theorem~\ref{thr:corrected}). Flow is from top to bottom. Each output can be reached by following consistent arrow color. To avoid misunderstanding, $S^{\prime \prime}$ is used instead of $S^\prime$ for Theorem~\ref{thr:mainImprovement2}. Both notations represent approximate defective sets recovered after running decoding algorithms. Set $S$ is the true defective set. Parameters $t(n, d - \ell, u; z], h(n, d - \ell, u; z]$, and $w$ are defined in Table~\ref{buice.t1}.}
\label{fig:Overview}
\end{figure*}

\subsection{Comparison}
\label{sub:cmp}

The one proposed theorem for the number of tests and three proposed non-adaptive algorithms are compared with previous ones in Table~\ref{buice.t1}. Our proposed algorithms are error-tolerant and their decoding algorithms are deterministic. Note that Ahlswede et al.~\cite{ahlswede2011bounds} also considered the case $d \neq \ell + u$; however, the bound on the number of tests has no constructive approximations for easy inference. Therefore, we do not include that bound in Table~\ref{buice.t1} for easy comparison.

\begin{table*}[t]

\begin{center}
\scalebox{.65}{
\begin{tabular}{|c|c|c|c|c|c|c|c|c|c|}
\hline
Scheme & \begin{tabular}{@{}c@{}} No. of \\ defectives \end{tabular} & \begin{tabular}{@{}c@{}} Thresholds \end{tabular} & \begin{tabular}{@{}c@{}} No. of \\items ($n$) \end{tabular} & \begin{tabular}{@{}c@{}} Model on \\ gap interval \end{tabular} & \begin{tabular}{@{}c@{}} Error \\tolerance \end{tabular} & \begin{tabular}{@{}c@{}} Number of tests \\ $t$ \end{tabular} & \begin{tabular}{@{}c@{}} Decoding time \\ (Decoding complexity) \end{tabular} & \begin{tabular}{@{}c@{}} Defective \\ set \\ recovered \end{tabular} & \begin{tabular}{@{}c@{}} Decoding \\ type \end{tabular} \\
\hline
\begin{tabular}{@{}c@{}} Ahlswede et al.~\cite{ahlswede2011bounds} \end{tabular} & $d = \ell + u$ & \begin{tabular}{@{}c@{}} $\ell < u \leq d$ \end{tabular} & $\geq d$ & No & $\times$ & \begin{tabular}{@{}c@{}} $O(u 2^{2u} \log{n})$ \end{tabular} & $\times$ & $\times$ & $\times$ \\
\hline
\begin{tabular}{@{}c@{}} Chen et al.~\cite{chen2008upper} \end{tabular} & $\leq d$ & \begin{tabular}{@{}c@{}} $\ell < u \leq d$ \end{tabular} & $\geq d$ & No & $z$ & \begin{tabular}{@{}c@{}} $t(n, d - \ell, u; z] = $ \\ $O \left( z \left( \frac{k}{u} \right)^u \left( \frac{k}{d - \ell} \right)^{d - \ell} k \ln{\frac{n}{k}} \right)$ \end{tabular} & $\times$ & $\times$ & $\times$ \\
\hline
\begin{tabular}{@{}c@{}} Cheraghchi~\cite{cheraghchi2013improved} \end{tabular} & $\leq d$ & \begin{tabular}{@{}c@{}} $\ell < u \leq d$ \end{tabular} & $\geq d$ & No & $O \left( \frac{p d^2 \log{\frac{n}{d}} }{(1 - p)^2} \right)$ & $O \left( \frac{d^{g + 2} \ln{\frac{n}{d}}}{(1-p)^2} \cdot (8u)^u \right)$ & $\times$ & $\times$  & $\times$ \\
\hline
\begin{tabular}{@{}c@{}} \textbf{Proposed 0} \\ (\textbf{Theorem~\ref{thr:improvedChenUpper}}) \end{tabular} & $\leq d$ & \begin{tabular}{@{}c@{}} $\ell < u \leq d$ \end{tabular} & $\geq \frac{(d + u)^2}{u}$ & \textbf{No} & $z$ & \begin{tabular}{@{}c@{}} $h(n, d - \ell, u; z]$ \\ $= O \left( \left( 1 + \frac{z}{\alpha} \right) \cdot \right.$ \\ $\left. \left( \frac{k}{u} \right)^u \left( \frac{k}{d - \ell} \right)^{d - \ell} k \ln{\frac{n}{k}} \right)$ \end{tabular} & $\times$ & $\times$ & $\times$ \\
\hline
\begin{tabular}{@{}c@{}} Chan et al.~\cite{chan2013stochastic} \end{tabular} & $d = o(n)$ & \begin{tabular}{@{}c@{}} $\ell < u = o(d)$ \end{tabular} & $\omega(d)$ & \begin{tabular}{@{}c@{}} Bernoulli \\ Linear \end{tabular} & $\times$ & \begin{tabular}{@{}c@{}} $O\left( \ln{\frac{1}{\epsilon}} \cdot d\sqrt{\ell} \ln{n} \right)$ \\ $O(g^2 d \ln{n} + d\ln{\frac{1}{\epsilon}})$  \end{tabular} & \begin{tabular}{@{}c@{}} $O(n\ln{n} + n \ln{\frac{1}{\epsilon}})$ \\ $O(g^2 n \ln{n} + n \ln{\frac{1}{\epsilon}})$ \end{tabular} & $S^\prime \equiv S$ & Rnd. \\
\hline
\begin{tabular}{@{}c@{}} Reisizadeh et al.~\cite{reisizadeh2018sub} \end{tabular} & \begin{tabular}{@{}c@{}} $d = O(n^\beta)$ \\ for \\$0 < \beta < 1$ \end{tabular} & \begin{tabular}{@{}c@{}} $\ell < u = o(d)$ \end{tabular} & $O(d^{1/\beta})$ & Bernoulli & $\times$ & $\Theta(\sqrt{u} d \ln^3{n} )$ & \begin{tabular}{@{}c@{}} $O(u^{1.5} d \ln^4{n} )$ \\ with a $O(u \ln{n}) \times \binom{n}{u}$ \\ look-up matrix \end{tabular} & $S^\prime \equiv S$ & Rnd. \\
\hline
\begin{tabular}{@{}c@{}} Chen and Fu~\cite{chen2009nonadaptive} \\ (\textbf{more accurate} in Theorem~\ref{thr:corrected}) \end{tabular} & $\leq d$ & \begin{tabular}{@{}c@{}} $\ell < u \leq d$ \end{tabular} & $\geq \frac{(d + u)^2}{u}$ & No & $z$ & $t(n, d - \ell, u; z]$ & \begin{tabular}{@{}c@{}} $O \left( t(n, d - \ell, u; z] \times u \left( \binom{n}{u} \right. \right.$ \\ $\left. \left. + (d - u) \binom{n - u}{g + 1} \binom{d - 1}{g} \binom{d}{u} \right) \right)$ \end{tabular} & \begin{tabular}{@{}c@{}} $|S^\prime \setminus S| \leq g$ \\ $|S \setminus S^\prime| \leq g$ \end{tabular} & Det. \\
\hline
\begin{tabular}{@{}c@{}} \textbf{Proposed 1} \\ (\textbf{Theorem~\ref{thr:improved0}}) \end{tabular} & $\leq d$ & $\ell < u \leq d$ & $\geq \frac{(d + u)^2}{u}$ & \textbf{No} & $z$ & \begin{tabular}{@{}c@{}} $h(n, d - \ell, u; z]$ \end{tabular} & \begin{tabular}{@{}c@{}} $O\left( h(n, d - \ell, u; z] \times u \left( \binom{n}{u} \right. \right.$ \\ $\left. \left. + (d - u) \binom{n - u}{g + 1} \binom{d - 1}{g} \binom{d}{u} \right) \right)$ \end{tabular} & \begin{tabular}{@{}c@{}} $|S^\prime \setminus S| \leq g$ \\ $|S \setminus S^\prime| \leq g$ \end{tabular} & \textbf{Det.} \\
\hline
\begin{tabular}{@{}c@{}} \textbf{Proposed 2} \\ (\textbf{Theorem~\ref{thr:mainImprovement}}) \end{tabular} & $\leq d$ & $\ell < u < d$ & $\geq \frac{\mathrm{e}^2(d + u)^2}{u}$ & \textbf{No} & $z$ & \begin{tabular}{@{}c@{}} $h(n, d - \ell, u; z]$ \end{tabular} & \begin{tabular}{@{}c@{}} $O \left( h(n, d - \ell, u; z] \cdot u\binom{n}{u} \right)$ \end{tabular} & \begin{tabular}{@{}c@{}} $|S^\prime \setminus S| \leq g w$ \\ $|S \setminus S^\prime| \leq g$ \end{tabular} & \textbf{Det.} \\
\hline
\begin{tabular}{@{}c@{}} \textbf{Proposed 3} \\ (\textbf{Theorem~\ref{thr:mainImprovement2}}) \end{tabular} & $\leq d$ & $\ell < u < d$ & $\geq \frac{\mathrm{e}^2 (d + u)^2}{u}$ & \textbf{No} & $z$ & \begin{tabular}{@{}c@{}} $h(n, d - \ell, u; z]$ \end{tabular} & \begin{tabular}{@{}c@{}} $O \left( h(n, d - \ell, u; z] \cdot u \cdot \left( \binom{n}{u} \right. \right.$ \\ $\left. \left. + (d - u) \binom{w + d - u}{g + 1} \binom{d - 1}{g} \binom{d}{u} \right) \right)$ \end{tabular} & \begin{tabular}{@{}c@{}} $|S^\prime \setminus S| \leq g$ \\ $|S \setminus S^\prime| \leq 2g$ \end{tabular} & \textbf{Det.} \\
\hline

\end{tabular}}

\end{center}

\caption{Comparison of proposed theorems with previous ones. A $\times$ symbol means that the criterion does not hold for that scheme. The terms ``Randomized'' and ``Deterministic'' are abbreviated to ``Rnd.'' and ``Det.''. Sets $S^\prime$ and $S$ are the recovered defective set and true defective set, respectively. We define $k = d - \ell + u$, $\alpha = k \ln{\frac{\mathrm{e} n}{k}} + u \ln{\frac{\mathrm{e} k}{u}}$, $w = \left( \left\lfloor |S|/(\ell + 1) \right\rfloor + u - 1 \right) g$, and $0 \leq p < 1$. Parameters $t(n, d - \ell, u; z]$ and $h(n, d - \ell, u; z]$ are defined in rows 2 and 4 as well as in~\eqref{eqn:ChenUpper} and~\eqref{eqn:improvedChenUpper}, respectively.}

\label{buice.t1}
\end{table*}

\subsubsection{Number of tests} When there are no models for the gap $g$, the upper bound on the number of tests with our proposed theorems is smaller than with the ones proposed by Chen and Fu~\cite{chen2009nonadaptive} and Chen et al.~\cite{chen2008upper}. Note that the upper bounds on the number of tests with Chen and Fu's scheme and Chen et al.'s scheme are equal, and so are our proposed theorems. The number of tests $O \left( \frac{d^{g + 2} \ln{\frac{n}{d}}}{(1-p)^2} \cdot (8u)^u \right)$ with the scheme proposed by Cheraghchi~\cite{cheraghchi2013improved} can be reduced to $O \left( \frac{d^{g + 2} \ln{\frac{n}{d}}}{(1-p)^2} \right)$ as $u$ is a constant; i.e., the multiplicity $(8u)^u$ can be removed because it is constant. It is essentially the optimal asymptotic number of tests. However, Cheraghchi~\cite{cheraghchi2013improved} does not focus on the finite length regime and refining the bounds for that as well as the algorithmic recovery problem. When $d = \ell + u$, a similar number of tests, which is $O(u 2^{2u} \log{n})$, is attained by Ahlswede et al.~\cite{ahlswede2011bounds}. The big $O$ notation is not useful in practice for this case because this multiplicity is extremely large and should not be removed. For example, we have $(8u)^u = 2^{20} = 1,048,576$ when $u = 4$ and $(8u)^u \geq 102,400,000$ when $u \geq 5$. Therefore, in terms of asymptotics, the number of tests with the scheme proposed by Cheraghchi is good as $u$ is constant, although it is extremely large in practice.

The number of tests could be significantly reduced by setting more conditions on $g, u$, and $d$, but such conditions would likely make any proposed scheme impractical. Moreover, the previous schemes that followed this approach do not take into account erroneous outcomes. When the Bernoulli model is applied to the gap, i.e., the number of defectives in a test is between the thresholds, the outcome is positive/negative with probability $0.5$. Setting $u = o(d)$ and error precision $\epsilon > 0$, Chan et al.~\cite{chan2013stochastic} achieved a small number of tests $O\left( \ln(1/\epsilon) \cdot d\sqrt{\ell} \ln{n} \right)$ while Reisizadeh et al.~\cite{reisizadeh2018sub} attained $\Theta(\sqrt{u} d \ln^3{n} )$ tests. When a linear model is applied to the gap, i.e., the number of defectives in a test is between the thresholds, the probability of a positive outcome linearly increases with the number of defectives. The number of tests with a linear model is $O(g^2 n \ln{n} + n \ln(1/\epsilon))$~\cite{chan2013stochastic}.

Once $g = 0$, D'yachkov et al.~\cite{d2013superimposed} and Cheraghchi~\cite{cheraghchi2013improved} show that it is possible to obtain an optimal bound on the number of tests, i.e., $O \left( d^2 \ln{n} \right)$ tests, when $u$ is a constant. Since the objective of this work is to consider the case $g > 0$, we recommend readers, who are interested in the case $g = 0$, to~\cite{bui2019efficiently} for further reading.

\subsubsection{Decoding time} Let $S^\prime$ and $S$ be the recovered defective set and the true defective set. For threshold group testing with gap $g$, $S^\prime$ and $S$ are indistinguishable if $|S^\prime \setminus S| \leq g$ and $|S \setminus S^\prime| \leq g$. Nevertheless, if a model is applied to the gap, $S^\prime \equiv S$ can be attained with some probability. With this approach, the fastest decoding was at with the scheme of Reisizadeh et al.~\cite{reisizadeh2018sub}: $O(u^{1.5} d \ln^4{n} )$. However, this scheme is based on the assumption that $\ell < u = o(d)$, that the Bernoulli model is applied to the gap, and that an auxiliary look-up matrix of size $O(u \ln{n}) \times \binom{n}{u}$ is stored somewhere. The need for a look-up matrix makes this scheme an impractical solution. For example, if $n = 10^6$ and $u = 5$, the number of columns in the look-up matrix is more than 8.3 octillion ($8.3 \times 10^{27}$). Moreover, $n$ and $u$ are more likely larger in practice. The scheme of Chan et al.~\cite{chan2013stochastic} attains a near-optimal decoding time: $O\left( \ln{\frac{1}{\epsilon}} \cdot d\sqrt{\ell} \ln{n} \right)$ or $O(g^2 d \ln{n} + d\ln{\frac{1}{\epsilon}})$ for $\epsilon > 0$. However, this decoding time is attained only under certain constraints: the Bernoulli or a linear model is applied to the gap, $n$ and $d = o(n)$ are large enough, and $\ell = o(d)$. This scheme is thus also likely impractical.

The conditions on the gap and on $n, \ell, u$, and $d$ make the schemes proposed by Chan et al.~\cite{chan2013stochastic} and Reisizadeh et al.~\cite{reisizadeh2018sub} impractical. Like Chen and Fu~\cite{chen2009nonadaptive}, we consider the case in which there are no constraints on the gap and $\ell < u \leq d < n$. Our decoding algorithms are deterministic. With the goal of attaining $|S^\prime \setminus S| \leq g$ and $|S \setminus S^\prime| \leq g$, the number of tests and the decoding time with our proposed algorithms (summarized in Theorems~\ref{thr:improved0},~\ref{thr:mainImprovement},~\ref{thr:mainImprovement2}) are much lower than the one proposed by Chen and Fu~\cite{chen2009nonadaptive} (summarized in Theorem~\ref{thr:corrected}). 

There are two terms in the decoding complexity of Theorem~\ref{thr:improved0} (in Proposed 1): $\binom{n}{u}$ and $(d - u) \binom{n - u}{g + 1} \binom{d - 1}{g} \binom{d}{u}$. To remove the second term, we relax the condition on $|S^\prime \setminus S|$ from $|S^\prime \setminus S| \leq g$ to $|S^\prime \setminus S| \leq wg$, where $w = \left( \left\lfloor \frac{|S|}{\ell + 1} \right\rfloor + u - 1 \right) g$. This reduces the decoding complexity of Theorem~\ref{thr:improved0} to $O\left( h(n, d - \ell, u; z] \times u \binom{n}{u} \right)$, which is significantly less than the original one in Theorem~\ref{thr:improved0}. This result is summarized in Theorem~\ref{thr:mainImprovement} (in Proposed 2).

However, it is clear that the condition $|S^\prime \setminus S| \leq wg$ in Theorem~\ref{thr:mainImprovement} is not as tight as the condition $|S^\prime \setminus S| \leq g$ in Theorem~\ref{thr:improved0}. To remedy this drawback, we derived Theorem~\ref{thr:mainImprovement2} (in Proposed 3), which slightly increases the decoding complexity while attaining the conditions $|S^\prime \setminus S| \leq 2g$ and $|S \setminus S^\prime| \leq g$.

\section{Preliminaries}
\label{sec:pre}

\subsection{Notations}
\label{sub:pre:notation}

For consistency, we use capital calligraphic letters for matrices, non-capital letters for scalars, bold letters for vectors, and capital letters for sets. All matrix and vector entries are binary. The frequently used notations are listed in Table~\ref{buice.t2}.

\begin{table}[t]

\begin{center}
\scalebox{1.0}{
\begin{tabular}{|c|l|}
\hline
\begin{tabular}{@{}c@{}} Notation \end{tabular} & \begin{tabular}{@{}c@{}} Description \end{tabular} \\
\hline
$n$ & Number of items \\
\hline
$d$ & Maximum number of defective items \\
\hline
$\bX = (x_1, \ldots, x_n)^T$ & Binary representation of $n$ items \\
\hline
$\ell$ & \begin{tabular}{@{}l@{}} Lower bound in \\ non-adaptive $(n, d, \ell, u)$-TGT model \end{tabular}  \\
\hline
$u$ & \begin{tabular}{@{}l@{}} Upper bound in \\ non-adaptive $(n, d, \ell, u)$-TGT model \end{tabular}  \\
\hline
$g = u - \ell - 1$ & Gap between $\ell$ and $u$ \\
\hline
$S = \{j_1, j_2, \ldots, j_{|S|} \}$ & \begin{tabular}{@{}l@{}} Set of defective items; \\ cardinality of $S$ is $|S| \leq d$ \end{tabular}  \\
\hline
$N = [n] = \{1, \ldots, n \}$ & Set of $n$ items \\
\hline
$\otimes_{\ell, u}$ & \begin{tabular}{@{}l@{}} Operation related to non-adaptive \\ $(n, d, \ell, u)$-TGT (to be defined later) \end{tabular}  \\
\hline
$\cT_{i, *}$ & Row $i$ of matrix $\cT$ \\
\hline
$\cT_{*, j}$ & Column $j$ of matrix $\cT$ \\
\hline
$\cM_{i,*}$ & Row $i$ of matrix $\cM$ \\
\hline
$\cM_{*,j}$ & Column $j$ of matrix $\cM$ \\
\hline
\end{tabular}}

\end{center}

\caption{Notations frequently used in this paper.}

\label{buice.t2}
\end{table}

\subsection{Problem definition}
\label{sub:pre:probDef}
We index the population of $n$ items from 1 to $n$. Let $[n] = \{1, 2, \ldots, n \}$ and $S$ be the defective set, where $|S| \leq d$. A test is defined by a subset of items $P \subseteq [n]$. A pool with a negative (positive) outcome is called a negative (positive) pool. The outcome of a test on a subset of items is positive if the subset contains at least $u$ defective items, is negative if the subset contains up to $\ell$ defective items, and arbitrary otherwise. Formally, the test outcome is positive if $|P \cap S| \geq u$, negative if $|P \cap S| \leq \ell$, and arbitrary if $\ell < |P \cap S| < u$. This model is denoted as $(n, d, \ell, u)$-TGT. In addition, $g = u - \ell - 1$ is the gap.

We can model non-adaptive $(n, d, \ell, u)$-TGT as follows. A $t \times n$ binary matrix $\cT =(t_{ij})$ is defined as a measurement matrix, where $n$ is the number of items and $t$ is the number of tests. Vector $\bX = (x_1,\ldots,x_n)^T$ is the binary representation vector of $n$ items, where $|\bX| = \sum_{i = 1}^n x_i \leq d$. An entry $x_j=1$ indicates that item $j$ is defective, and $x_j=0$ indicates otherwise. The $j$th item corresponds to the $j$th column of the matrix. An entry $t_{ij}=1$ naturally means that item $j$ belongs to test $i$, and $t_{ij}=0$ means otherwise. The outcome of all tests is $\bY=(y_1, \ldots, y_t)^T$, where $y_i=1$ if test $i$ is positive and $y_i=0$ otherwise. The procedure used to get outcome vector $\bY$ is called \textit{encoding}. The procedure used to identify defective items from $\bY$ is called \textit{decoding}. Outcome vector $\bY$ is given by

\begin{equation}
\label{eqn:thresholdGT}
\bY = \cT \otimes_{\ell, u} \bX = \begin{bmatrix}
\cT_{1, *} \otimes_{\ell, u} \bX \\
\vdots \\
\cT_{t, *} \otimes_{\ell, u} \bX
\end{bmatrix} = \begin{bmatrix}
y_1 \\
\vdots \\
y_t
\end{bmatrix},
\end{equation}
where $\otimes_{\ell, u}$ is a notation for the test operation in non-adaptive $(n, d, \ell, u)$-TGT; namely, $y_i = \cT_{i, *} \otimes_{\ell, u} \bX = 1$ if $\sum_{j=1}^n x_j t_{ij} \geq u$, $y_i = \cT_{i, *} \otimes_{\ell, u} \bX = 0$ if $\sum_{j=1}^n x_j t_{ij} \leq \ell$, and $y_i = \cT_{i, *} \otimes_{\ell, u} \bX = \{0, 1\}$ if $\ell < \sum_{j=1}^n x_j t_{ij} < u$, for $i=1, \ldots, t$.

Our objective is to find an efficient encoding and decoding scheme with non-adaptive approach to identify up to $d$ defective items in non-adaptive $(n, d, \ell, u)$-TGT. Precisely, our task is to minimize the number of rows in matrix $\cT$ and the time for recovering $\bX$ from $\bY$ by using $\cT$.

\subsection{Disjunct matrices}
\label{sub:pre:disjunct}

Disjunct matrices are a powerful tool to tackle the threshold group testing problem~\cite{chen2009nonadaptive,cheraghchi2013improved,bui2019efficiently}. They were first introduced by Kautz and Singleton~\cite{kautz1964nonrandom} as \textit{superimposed codes} and then generalized by Stinson and Wei~\cite{stinson2004generalized} and D'yachkov et al.~\cite{d2002families}. The support set for vector $\bV = (v_1, \ldots, v_w)$ is $\supp(\bV) = \{j \mid v_j \neq 0 \}$. The formal definition of a disjunct matrix is as follows.

\begin{definition}
An $m \times n$ binary matrix $\cM$ is called an $(n, d, r; z]$-disjunct matrix if, for any two disjoint subsets $S_1, S_2 \subset [n]$ such that $|S_1| = d$ and $|S_2| = r$, there exists at least $z$ rows in which there are all 1's among the columns in $S_2$ while all the columns in $S_1$ have 0's, i.e., $\left\vert \bigcap_{j \in S_2} \supp \left( \cM_{*, j} \right) \big\backslash \bigcup_{j \in S_1} \supp \left( \cM_{*, j} \right) \right\vert \geq z$. Parameter $\left\lfloor (z-1)/2 \right\rfloor$ is usually referred to as the \textit{error tolerance}.
\label{def:threshDisjunct}
\end{definition}

Matrix $\cM$ can be illustrated as follows.

\begin{align}
\cM = 
\left[
\begin{array}{c}
\ldots\\
\ldots \\
\ldots \\
\ldots \\
\ldots \\
\ldots \\
\end{array} \right.
\overbrace{
\begin{array}{cc}
\ldots & \ldots \\
1 & 1 \\
\ldots & \ldots \\
1 & 1 \\
\ldots & \ldots \\
\ldots & \ldots \\
\end{array}}^{r}
\begin{array}{c}
\ldots \\
\ldots \\
\ldots \\
\ldots \\
\ldots \\
\ldots \\
\end{array}
\overbrace{
\begin{array}{cc}
\ldots & \ldots \\
0 & 0 \\
\ldots & \ldots \\
0 & 0 \\
\ldots & \ldots \\
\ldots & \ldots \\
\end{array}}^{d}
\left.
\begin{array}{c}
\ldots \\
\ldots \\
\ldots \\
\ldots \\
\ldots \\
\ldots \\
\end{array}
\right]
\begin{array}{c}
\\
\begin{tabular}{@{}l@{}}the $1$st \\specific row \end{tabular} \\
\\
\begin{tabular}{@{}l@{}}the $z$th \\specific row \end{tabular} \\
\\
\\
\end{array}
\nonumber
\end{align}

Chen et al.~\cite{chen2008upper} gave an upper bound on the number of rows for $(n, d, u; z]$-disjunct matrices as follows.

\begin{theorem}~\cite[Theorem 3.2]{chen2008upper}
\label{thr:ChenUpper}
For any positive integers $d, u, z$, and $n$ with $k = d + u \leq n$, there exists a $t \times n$ $(n, d, u; z]$-disjunct matrix with
\begin{align}
t(n, d, u; z] &= z \left( \frac{k}{u} \right)^u \left( \frac{k}{d} \right)^d \left[ 1 + k \left(1 + \ln{\left( \frac{n}{k} + 1 \right)} \right) \right] \nonumber \\
&= O \left( z \left( \frac{k}{u} \right)^u \left( \frac{k}{d} \right)^d k \ln{\frac{n}{k}} \right) \nonumber \\
&= O(z \cdot t(n, d, u; 1]). \label{eqn:ChenUpper}
\end{align}
\end{theorem}

\section{Review and analysis of Chen and Fu's work}
\label{sec:ReviewAndImprovement}

\subsection{Preliminaries}

To clarify the basis of our proposed algorithms, we review Chen and Fu's work~\cite{chen2009nonadaptive} which is the first and only work tackling the ``for all'' model in NATGT with a gap and with a decoding algorithm. They proposed schemes for finding the defective items using $t(n, d - \ell, u; z]$ tests in time $O(n^u \ln{n})$. However, the decoding time becomes impractical as $n$ or $u$ increases. The intuition of Chen and Fu algorithm is to initialize an approximate $S^\prime$ of size $u$ such that the outcome of the test on $S^\prime$ is positive. The algorithm then proceeds to increase the size of $S^\prime$ such that the cardinality of $S^\prime$ is not larger than the maximum number of defectives, i.e., $d$, and the outcome of a test on every subset of $u$ items in $S^\prime$ is positive.

To facilitate the problem of identifying defectives, the graph search problem is first introduced. Given a vertex set $V = \{1, \ldots, n \}$, the goal is to reconstruct a hidden graph $\sfH$ defined on $V$ by asking queries in the following format: for $U \subseteq V$, the query is ``Does a complete graph induced by $U$ contain any edge of $\sfH$?'' In other words, a pool containing all vertices in $U$ is positive if at least one edge of $\sfH$ is also an edge of the complete graph induced by $U$.

Given a finite set $V$, a hypergraph $\mathbb{H} = (V, \sfF)$ is a family $\sfF = \{ E_1, E_2, \ldots, E_m\}$ of subsets of $V$. The elements of $V$ are called vertices, and the subsets $E_i$'s are the edges of the hypergraph $\mathbb{H}$.

A hypergraph is called a $u$-hypergraph if each edge consists of exactly $u$ vertices. A subset of a set is called a $u$-subset if it contains exactly $u$ elements of the set. Let $W$ be a subset of $V$. A hypergraph is $u$-complete with respect to $W$ if and only if (iff) every $u$-subset of $W$ is an edge of the hypergraph.

Recall that our objective is to identify a set of defectives $S$ from a given set of items $N = [n]$. Let $S^\prime$ be a set such that $|S^\prime \setminus S| \leq g$ and $|S \setminus S^\prime| \leq g$. Note that there is more than one set $S^\prime$ satisfying these properties. Let $[n] = \{1, 2, \ldots, n \}$ be vertex set $V$. Suppose that a set of edges $\sfF$ contains all $u$-subsets of $S$ and a fraction of all or all $u$-subsets of every $S^\prime$. We can convert threshold group testing with a gap into the problem of reconstructing a hidden graph $\sfH$ in $\mathbb{H} = (V, \sfF)$ that is $u$-complete with respect to some $S^\prime$.

\subsection{Main idea}
\label{sub:review:main}

The main idea is to construct a family $\sfF$ such that, for any subset $X \in \sfF$, $|X| = u$, $|X \cap S| \geq \ell + 1$ and every $u$-subset $X^+ \subseteq S$ must be in $\sfF$. An approximate defective set $S^\prime$ is then recovered by using $\sfF$, where $|S^\prime \setminus S| \leq g$ and $|S \setminus S^\prime| \leq g$. Note that $S^\prime$ is the best defective set that can be recovered~\cite{damaschke2006threshold}.

To construct $\sfF$, an indicator of ``false negatives'' is introduced. We say that a set $X$ of the columns in a matrix appears in a row if every column in $X$ has a 1 in the row. For a subset $X$ of the columns in matrix $\cM$, we define $t_0^{\cM}(X)$ to be the number of negative pools in which all columns in $X$ appear. Attaining $S^\prime$ is done by increasing the size of an approximate defective set $S^\prime$ from $u$ until the properties $|S^\prime \setminus S| \leq g$ and $|S \setminus S^\prime| \leq g$ hold. In other words, the number of false positives and false negatives are up to $g$.

Given measurement matrix $\cM$, Chen and Fu supposed that $\bY$ is the outcome vector with up to $e$ erroneous outcomes in non-adaptive $(n, d, \ell, u)$-TGT. By setting $\cM$ as an $(n, d - \ell, u; 2e + 1]$-disjunct matrix, the authors obtained a decoding algorithm in which an approximate set $S^\prime$ is attained, as shown in Algorithm~\ref{alg:decodingThreshold}. Step~\ref{alg:createHypergraph} is to construct a family $\sfF$ and a hypergraph $\mathbb{H} = (V, \sfF)$. Step~\ref{alg:getDefectiveSet} is to attain $S^\prime$ by using $\mathbb{H}$, as illustrated in Fig.~\ref{fig:ChenFu}. More precisely, the algorithm first initializes set $S_1$ consisting of the $u$ vertices belonging to an edge of the family $\sfF$. A new set $S_{i + 1}$ is then created such that $|S_{i + 1}| = |S_i| + 1$. Set $S_{i + 1}$ is made equal to set $(S_i \cup A_i) \setminus B_i$ by selecting set $A_i$ of $g + 1$ elements in $V \setminus S_i$ and set $B_i$ of $g$ elements in $S_i$ such that $\mathbb{H}$ is $u$-complete with respect to $(S_i \cup A_i) \setminus B_i$. It is obvious that $|S_{i + 1}| = |S_i| + 1$. This process stops once either $S_i$ is not extendable or $|S_i| \geq d$. If the process stops when $i = m$, $S^\prime$ is set to $S_m$.

\begin{figure*}[ht]
\centering
\includegraphics[scale=0.5]{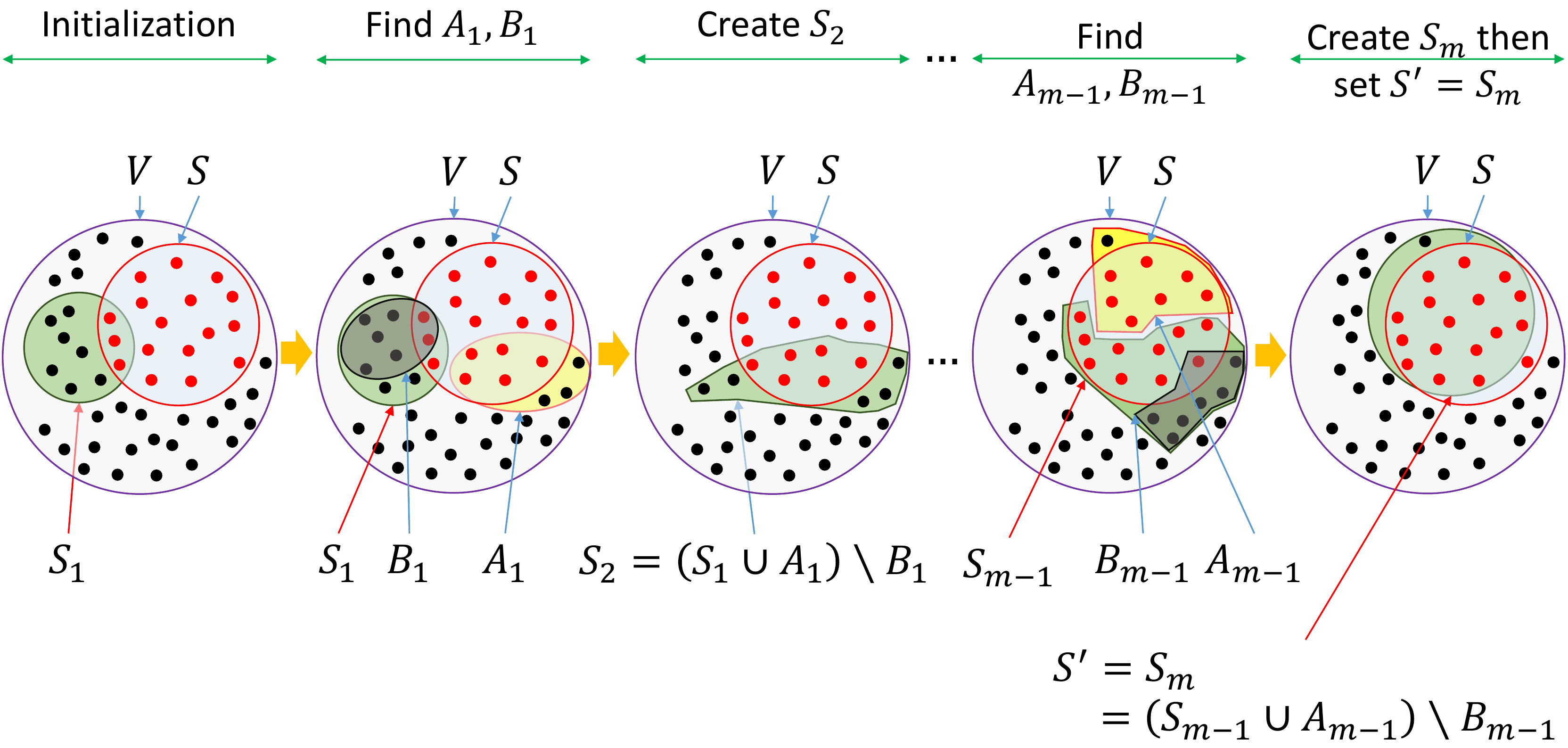}
\caption{Illustration of finding an approximate defective set $S^\prime$ of the defective set $S$ such that $|S^\prime \setminus S| \leq g$ and $|S \setminus S^\prime| \leq g$ for Algorithm~\ref{alg:decodingThreshold}. We set $g = 7, u = 10$, and $\ell = u - g - 1 = 2$.}
\label{fig:ChenFu}
\end{figure*}

\begin{algorithm}
\caption{[Algorithm 2~\cite{chen2009nonadaptive}] $\mathrm{Decoding}_1(\bY, \cM)$: Decoding procedure for non-adaptive $(n, d, \ell, u)$-TGT with up to $e$ erroneous outcomes.}
\label{alg:decodingThreshold}
\textbf{Input:} Outcome vector $\bY$, a $(n, d - \ell, u; z = 2e + 1]$-disjunct matrix $\cM$.\\
\textbf{Output:} Set of defective items $S^\prime$ s.t. $|S^\prime \setminus S| \leq g$ and $|S \setminus S^\prime| \leq g$.

\begin{algorithmic}[1]
\State Construct a hypergraph $\mathbb{H} = (V, \sfF)$, where $V = [n]$ is the vertex set of $n$ items and a $u$-subset $X \subseteq [n]$ is an edge in $\sfF$ iff $t_0^{\cM}(X) \leq e$. \label{alg:createHypergraph}
\State We want to establish increasing vertex-sets $S_i$'s, $|S_1| < |S_2| \ldots < |S_m|$, such that the hypergraph $\mathbb{H}$ is $u$-complete with respect to each $S_i$. As an initial $S_1$, we may choose all $u$ vertices of an arbitrary edge. To find $S_{i + 1}$ for $i \geq 1$, we check all possible cases to obtain some $(g + 1)$-subset $A_i$ in $V(\mathbb{H}) \setminus S_i$ and a $g$-subset $B_i$ in $S_i$ such that $\mathbb{H}$ is $u$-complete with respect to $(S_i \cup A_i) \setminus B_i$. If such a pair $A_i, B_i$ exists, then set $S_{i + 1} = (S_i \cup A_i) \setminus B_i$. Continue this process till either $S_m$ is not extendable or $|S_i| \geq d$. Output the set $S^\prime = S_m$. \label{alg:getDefectiveSet}
\end{algorithmic}
\end{algorithm}

By using an $(n, d - \ell, u; z = 2e + 1]$-disjunct matrix and Algorithm~\ref{alg:decodingThreshold}, we can attain an approximate defective set $S^\prime$ as follows.

\begin{theorem}~\cite[Theorem 4.4]{chen2009nonadaptive}
\label{thr:ChenDecoding}
For an $(n, d, \ell, u)$-TGT model with at most $e$ erroneous outcomes, there exists a non-adaptive algorithm that successfully identifies some set $S^\prime$ with $|S^\prime \setminus S| \leq g$ and $|S \setminus S^\prime| \leq g$, using no more than $t(n, d - \ell, u; z = 2e + 1]$ tests. Moreover, the decoding complexity is 
\begin{align}
&t(n, d - \ell, u; z] \times u \binom{n}{u} + (d - u) \binom{n - u}{g + 1} \binom{d - 1}{g} \binom{d}{u} \label{eq:ChenUpper} \\
&= O \left( z \left( \frac{k}{u} \right)^u \left( \frac{k}{d - \ell} \right)^{d - \ell} k \ln{\frac{n}{k}} \cdot u \binom{n}{u} \right), \nonumber
\end{align}
where $k = d - \ell + u$.
\end{theorem}

The complexity of the theorem above is attained by taking the sum of the complexities of Steps~\ref{alg:createHypergraph} and~\ref{alg:getDefectiveSet}. Step~\ref{alg:createHypergraph} is done in time $t(n, d - \ell, u; z] \times u \binom{n}{u}$. Step~\ref{alg:getDefectiveSet} is done in time $(d - u) \binom{n - u}{g + 1} \binom{d - 1}{g} \binom{d}{u}$, which is \textit{inaccurate} in general. A detailed analysis is given in the Appendix. Here we present a more accurate version of Theorem~\ref{thr:ChenDecoding}.

\begin{theorem}[A more accurate version of Theorem 4.4 in\cite{chen2009nonadaptive}]
\label{thr:corrected}
For an $(n, d, \ell, u)$-TGT model with at most $e$ erroneous outcomes, there exists a non-adaptive algorithm that successfully identifies some set $S^\prime$ with $|S^\prime \setminus S| \leq g$ and $|S \setminus S^\prime| \leq g$ using no more than $t(n, d - \ell, u; z = 2e + 1]$ tests. Moreover, the decoding complexity is 
\begin{align}
&O \left( t(n, d - \ell, u; z] \times u \left( \binom{n}{u} \right. \right. \nonumber \\
& \left. \left. + (d - u) \binom{n - u}{g + 1} \binom{d - 1}{g} \binom{d}{u} \right) \right) \label{eq:corrected} \\
=& O \left( z \left( \frac{k}{u} \right)^u \left( \frac{k}{d - \ell} \right)^{d - \ell} k \ln{\frac{n}{k}} \right. \nonumber \\
&\left. \times u \left( \binom{n}{u} + (d - u) \binom{n - u}{g + 1} \binom{d - 1}{g} \binom{d}{u} \right) \right), \nonumber
\end{align}
where $k = d - \ell + u$.
\end{theorem}

\subsection{Example for Algorithm~\ref{alg:decodingThreshold}}
\label{sub:review:eg}

In this section, we demonstrate Algorithm~\ref{alg:decodingThreshold} by setting $n = 6, d = 4, \ell = 0, u = 2$, and $z = 1.$ This means that $g = u - \ell - 1 = 1$ and $e = 0$. We assume that the defective items are 1, 2, 4, and 5; i.e., the input vector is $\bX = (1, 1, 0, 1, 1, 0)^T$. The true defective set is therefore $S = \{1, 2, 4, 5 \} = \supp(\bX)$. The $(n = 6, d - \ell = 4, u = 2; z = 1]$-disjunct matrix $\cM$ is as follows:

\begin{equation}
\cM = \left[
\begin{tabular}{cccccc}
1 & 1 & 0 & 0 & 0 & 0 \\
1 &	0 &	1 &	0 &	0 &	0 \\
1 &	0 &	0 &	1 &	0 &	0 \\
1 &	0 &	0 &	0 &	1 &	0 \\
1 &	0 &	0 &	0 &	0 &	1 \\
0 &	1 &	1 &	0 &	0 &	0 \\
0 &	1 &	0 &	1 &	0 &	0 \\
0 &	1 &	0 &	0 &	1 &	0 \\
0 &	1 &	0 &	0 &	0 &	1 \\
0 &	0 &	1 &	1 &	0 &	0 \\
0 &	0 &	1 &	0 &	1 &	0 \\
0 &	0 &	1 &	0 &	0 &	1 \\
0 &	0 &	0 &	1 &	1 &	0 \\
0 &	0 &	0 &	1 &	0 &	1 \\
0 &	0 &	0 &	0 &	1 &	1 \\
1 & 0 & 1 & 0 & 1 & 0 \\
0 & 1 & 1 & 0 & 1 & 1 \\
0 & 0 & 1 & 1 & 0 & 1 \\
1 & 0 & 1 & 1 & 0 & 1 \\
1 & 0 & 1 & 0 & 0 & 1
\end{tabular}
\right], \bY = \cM \otimes_{0, 2} \bX = \left[ \begin{tabular}{c}
1 \\
1 \\
1 \\
1 \\
0 \\
1 \\
1 \\
1 \\
0 \\
1 \\
1 \\
0 \\
1 \\
0 \\
1 \\
1 \\
1 \\
0 \\
1 \\
0
\end{tabular}
\right]. \label{eqn:exChenFu}
\end{equation}

We assume that the observed vector is $\bY$, as in~\eqref{eqn:exChenFu}. Algorithm~\ref{alg:decodingThreshold} proceeds as follows. In Step~\ref{alg:createHypergraph}, hypergraph $\mathbb{H} = (V, \sfF)$ is constructed with the set of vertexes $V = [6] = \{1, 2, 3, 4, 5, 6 \}$. A search is made for all 2-subsets $X \in V$ in order to form the set of edges $\sfF$ such that $t^{\cM}_0(X) \leq e = 0$. From~\eqref{eqn:exChenFu}, we get $\sfF = \{ \{1, 2\}, \{1, 4\}, \{1, 5\}, \{2, 3\}, \{2, 4\}, \{2, 5\}, \{3, 5\}, \{4, 5\}$, $\{5, 6\} \}$\footnote{We delineate this example to ensure understanding. Since we use $u = 2$, hypergraph $\mathbb{H}$ becomes a normal graph in which an edge consists of two vertices. However, once $u \geq 3$, an edge in $\mathbb{H}$ consists of at least three vertices. Graph $\mathbb{H}$ is thus no longer a normal graph.}.

Step~\ref{alg:getDefectiveSet} starts with an initial $2$-subset $S_1 = \{1, 2\}$. All possible cases are checked to obtain some 2-subset $A_1$ ($g + 1 = 2$) in $V \setminus S_1 = \{ 3, 4, 5, 6\}$, which is some element of $\{ \{3, 4\}, \{3, 5\}, \{3, 6\}, \{4, 5\}, \{4, 6\}, \{5, 6\} \}$, and a 1-subset $B_1$ ($g = 1$) in $S_1$, which is some element of $\{ \{1 \}, \{2 \} \}$, such that $\mathbb{H}$ is $2$-complete with respect to $(S_1 \cup A_1) \setminus B_1$. Since $A_1 = \{3, 5\}$ and $B_1 = \{1 \}$ ensure that the condition holds, set $S_2 = (S_1 \cup A_1) \setminus B_1 = \{2, 3, 5 \}$.

Since $|S_2| = 3 < 4 = d$, we continue Step~\ref{alg:getDefectiveSet} by choosing a 2-subset $A_2 \subseteq V \setminus S_2 = \{1, 4, 6 \}$ and a 1-subset $B_2 \subseteq S_2$. The lists of potential subsets for $A_2$ and $B_2$ are $\{ \{1, 4 \}, \{1, 6 \}, \{4, 6 \} \}$ and $\{ \{2 \}, \{3 \}, \{5 \} \}$, respectively. We choose $A_2 = \{1, 4 \}$ and $B_2 = \{3 \}$ because $\mathbb{H}$ is $2$-complete with respect to $(S_2 \cup A_2) \setminus B_2$. Set $S_3 = (S_2 \cup A_2) \setminus B_2 = \{1, 2, 4, 5 \}$. Since $|S_3| = 4 \geq 4 = d$, the algorithm stops and output $S^\prime = S_3 = \{1, 2, 4, 5 \}$. In this case, the approximate defective set $S^\prime$ is identical to the true defective set $S$.

\section{Improved upper bounds on the number of tests for disjunct matrix}
\label{main:sec:improvedDisjunct}

In this section, we present better exact upper bounds on the number of tests compared to the one in Theorem~\ref{thr:ChenUpper}.

\subsection{First result}
\label{main:sub:improvedDisjunct:first}

The upper bound on the number of tests with Theorem~\ref{thr:ChenUpper} is large because of the multiplicity $z$. We present a better upper bound on the number of tests as follows.

\begin{theorem}
\label{thr:improvedChenUpper}
Let $2 \leq u \leq d < k = d + u \leq n$ be integers with $(d + u)^2/u \leq n$. Set $\alpha = k \ln{\frac{\mathrm{e} n}{k}} + u \ln{\frac{\mathrm{e} k}{u}}$. For any positive integer $z$, there exists an $h \times n$ $(n, d, u; z]$-disjunct matrix with
\begin{equation}
h(n, d, u; z] = O \left( \left( 1 + \frac{z}{\alpha} \right) \cdot \left( \frac{k}{u} \right)^u \left( \frac{k}{d} \right)^d k \ln{\frac{n}{k}} \right). \label{eqn:improvedChenUpper}
\end{equation}
\end{theorem}

\begin{proof}
Consider a randomly generated $h \times n$ matrix $\cG = (g_{ij})_{1 \leq i \leq h, 1 \leq j \leq n}$ in which each entry $g_{ij}$ is assigned to 1 with probability $p$ and to 0 with probability $1 - p$. For any pair of disjoint subsets $S_1, S_2 \subset [n]$ such that $|S_1| = u$ and $|S_2| = d$, we denote the event that for a row, there are 1's among the columns in $S_1$ and all 0's among the columns in $S_2$ on the same row by \textit{a good event}. The probability that the good event happens is:
\begin{equation}
q = p^u (1 - p)^d. \label{eqn:1row}
\end{equation}

Set $\alpha = k \ln{\frac{\mathrm{e} n}{k}} + u \ln{\frac{\mathrm{e} k}{u}}$ and $\beta = 1 - 2/\alpha$. It is obvious that $0 < \alpha, \beta$. We then set $z = (1 - \delta) q h$, where $0 < \delta < 1$. We will later prove that there always exists $\delta$ which depends on $n, u, d,$ and $z$ such that $z = (1 - \delta) q h$. For a pair of disjoint subsets $S_1, S_2 \subset [n]$ such that $|S_1| = u$ and $|S_2| = d$, let $\sfX_i = 1$ be an event that a good event occurs at row $i$ and $\sfX_i = 0$ be an event that a good event does not occur at row $i$. It is obvious that $\Pr[\sfX_i = 1] = q$, $\Pr[\sfX_i = 0] = 1 - q$, and $E[\sfX_i] = q$. Let $\sfX = \sum_{i = 1}^h \sfX_i$ denote the number of the good events happen for $h$ rows. We get $\mu = E[\sfX] = \sum_{i = 1}^h E[\sfX_i] = q h$.

By using Chernoff's bound, for fixed $S_1$ and $S_2$, the probability that a good event occurs for up to $z$ rows among $h$ rows is
\begin{align}
\Pr[\sfX \leq z] &= \Pr[\sfX \leq (1 - \delta) \mu] \nonumber \\
&\leq \exp \left( - \frac{\delta^2 \mu}{2} \right) = \exp \left( - \frac{\delta^2 q h}{2} \right). \nonumber
\end{align}

Using a union bound, the expected value of the number of good events in which each good event occurs for no more than $z$ rows among $h$ rows for all disjoint subsets $S_1, S_2 \subset [n]$ with $|S_1| = u$ and $|S_2| = d$, i.e., the probability that $\cG$ is not an $(n, d, u; z]$-disjunct matrix, is at most
\begin{align}
g(p, h, u, d, n) &= \binom{n}{d + u} \binom{d + u}{u} \Pr[\sfX \leq z] \nonumber \\
&\leq \binom{n}{k} \binom{k}{u} \exp \left( - \frac{\delta^2 q h}{2} \right). \label{eqn:2Func}
\end{align}

To ensure the existence of an $(n, d, u, g; z]$-disjunct matrix $\cG$, one needs to find $p$ and $h$ such that $g(p, h, u, d, n) < 1$. Set $p = \frac{u}{d + u} = \frac{u}{k}$ and $q =  p^u (1 - p)^{d} = \left( \frac{u}{k} \right)^u \left( \frac{d}{k} \right)^{d}$. We then have

\begin{equation}
g(p, h, u, d, n) \leq \binom{n}{k} \binom{k}{u} \exp \left( - \frac{\delta^2 q h}{2} \right) < 1. \nonumber
\end{equation}

For this to hold, it suffices that
\begin{alignat}{3}
&& \binom{n}{k} \binom{k}{u} &\leq \left( \frac{\mathrm{e} n}{k} \right)^{k} \left( \frac{\mathrm{e} k}{u} \right)^{u} < \exp \left( \frac{\delta^2 q h}{2} \right) \label{eqn:Exist2} \\
\Longleftrightarrow&& h &> \frac{2}{\delta^2} \cdot \frac{1}{q} \cdot \left( k \ln{\frac{\mathrm{e} n}{k}} + u \ln{\frac{\mathrm{e} k}{u}}  \right) \nonumber \\
\Longleftrightarrow&& &> \frac{2}{\delta^2} \cdot \left( \frac{k}{u} \right)^u \left( \frac{k}{d} \right)^d \left( k \ln{\frac{\mathrm{e} n}{k}} + u \ln{\frac{\mathrm{e} k}{u}}  \right). \label{eqn:Exist3}
\end{alignat}

In the above, we have \eqref{eqn:Exist2} because $\binom{a}{b} \leq \left( \frac{\mathrm{e} a}{b} \right)^b$. Since $p = \frac{u}{k}$, from \eqref{eqn:Exist3}, if we set 
\begin{align}
h &= h(n, d, u; z] \nonumber \\
&= \frac{3}{\delta^2} \cdot \frac{1}{q} \cdot \left( k \ln{\frac{\mathrm{e} n}{k}} + u \ln{\frac{\mathrm{e} k}{u}}  \right) \nonumber \\
&= \frac{3}{\delta^2} \cdot \frac{1}{q} \cdot \alpha, \mbox{ where } \alpha = k \ln{\frac{\mathrm{e} n}{k}} + u \ln{\frac{\mathrm{e} k}{u}}, \label{eqn:calZ} \\
&= \frac{3}{\delta^2} \cdot \left( \frac{k}{u} \right)^u \left( \frac{k}{d} \right)^d \cdot \left( k \ln{\frac{\mathrm{e} n}{k}} + u \ln{\frac{\mathrm{e} k}{u}} \right), \nonumber
\end{align}
then $g(p, h, u, w, n) < 1$; i.e., there exists an $(n, d, u; z]$-disjunct matrix of size $h \times n$.

We now calculate $\delta$ versus $n, d, u$, and $z$. Since $z = (1 - \delta) qh$ and $h = \frac{3}{\delta^2} \cdot \frac{1}{q} \cdot \alpha$ in~\eqref{eqn:calZ}, we have:
\begin{align}
&& z &= (1 - \delta) qh = (1 - \delta) \cdot \frac{3\alpha}{\delta^2} \\
&\Longleftrightarrow& z \delta^2 + 3\alpha \delta - 3\alpha &= 0
\end{align}

Since the left side is a quadratic equation of $\delta$ and $\delta > 0$, we can derive
\begin{equation}
\delta = \frac{-3\alpha + \sqrt{9\alpha^2 + 12 \alpha z}}{2z} = \frac{\sqrt{3\alpha} \left( \sqrt{3\alpha + 4z} - \sqrt{3 \alpha} \right) }{2z}. \label{eqn:Delta}
\end{equation}

Let $f(x) = \sqrt{x}$. We have $f(x)$ is continuous on a closed interval $[3 \alpha, 3 \alpha + 4z]$ and differentiable on the open interval $(3 \alpha, 3 \alpha + 4z)$. By using the Lagrange's mean value theorem, then there is at least one point $b \in (3 \alpha, 3 \alpha + 4z)$ such that
\begin{align}
f(3\alpha + 4z) - f(3 \alpha) &= \sqrt{3\alpha + 4z} - \sqrt{3 \alpha} \nonumber \\
&= 4z \cdot f^\prime(b) = 4z \cdot \frac{1}{2 \sqrt{b}} = \frac{2z}{\sqrt{b}}.
\end{align}

Combine with~\eqref{eqn:Delta}, we get
\begin{align}
\delta = \frac{\sqrt{3\alpha} \left( \sqrt{3\alpha + 4z} - \sqrt{3 \alpha} \right) }{2z} = \frac{\sqrt{3\alpha} }{2z} \cdot \frac{2z}{\sqrt{b}} = \sqrt{\frac{3 \alpha}{b}}.
\end{align}

Because $b \in (3 \alpha, 3 \alpha + 4z)$, the following condition is straightforwardly attained
\begin{equation}
\frac{1}{\delta^2} = \frac{b}{3 \alpha} \in \left( 1, 1 + \frac{4z}{3 \alpha} \right).
\end{equation}

Therefore, the number of tests required is
\begin{align}
h &= h(n, d, u; z] \nonumber \\
&= \frac{3}{\delta^2} \cdot \left( \frac{k}{u} \right)^u \left( \frac{k}{d} \right)^d \cdot \left( k \ln{\frac{\mathrm{e} n}{k}} + u \ln{\frac{\mathrm{e} k}{u}} \right) \nonumber \\
&< 3 \left( 1 + \frac{4z}{3 \alpha} \right) \cdot \left( \frac{k}{u} \right)^u \left( \frac{k}{d} \right)^d \cdot \left( k \ln{\frac{\mathrm{e} n}{k}} + u \ln{\frac{\mathrm{e} k}{u}} \right). \nonumber
\end{align}

\end{proof}

Since $\alpha$ is always larger than $4/3$, $4z/(3\alpha)$ is always smaller than $z$. It implies that the upper bound on the number of tests in Theorem~\ref{thr:improvedChenUpper} is always tighter than the one in Theorem~\ref{thr:ChenUpper}.

\textbf{Discussion of number of tests for TGT and CGT:} With the same settings for $n, d$, and the maximum number of erroneous outcomes $\lfloor (z-1)/2 \rfloor$, what are the similarities and differences for the number of tests between TGT and CGT? We first transform~\eqref{eqn:improvedChenUpper}:
\begin{equation}
h(n, d, u; z] = t(n, d, u; 1] +  O \left( \frac{z}{\alpha} \right) \cdot t(n, d, u; 1]. \nonumber
\end{equation}

In Corollary 19~\cite{ngo2011efficiently}, Ngo, Porat, and Rudra show that the number of tests needed to handle $\lfloor (z-1)/2 \rfloor$ erroneous outcomes is $O(d^2 \log{n}) + O(zd) = t(n, d, 1; 1] + O \left( \frac{z}{d \log{n}} \right) \cdot t(n, d, 1; 1]$. It is well known that $t(n, d, 1; 1] = O(d^2 \log{n})$ is the achievable bound on the number of tests for the noiseless setting ($z = 1$). The authors prove that we only need $O \left( \frac{z}{d \log{n}} \right) \cdot t(n, d, 1; 1]$ additional tests to handle up to $\lfloor (z-1)/2 \rfloor$ erroneous outcomes instead of using $z \times t(n, d, 1; 1]$. The result for Theorem~\ref{thr:improvedChenUpper} shares this property. Since $t(n, d, u; 1]$ is the achievable number of tests for the noiseless setting in TGT, we need only $O \left( \frac{z}{\alpha} \right) \cdot t(n, d, u; 1]$ additional tests to handle up to $\lfloor (z-1)/2 \rfloor$ erroneous outcomes instead of $z \times t(n, d, u; 1]$ as in Theorem~\ref{thr:ChenUpper}.

\subsection{Second result}
\label{main:sub:improvedDisjunct:second}

With an addition constraint on $z$, an alternative version of Theorem~\ref{thr:improvedChenUpper} can be derived to directly attain a better upper bound on the number of tests compared with the upper bound in Theorem~\ref{thr:ChenUpper}.

\begin{theorem}
\label{thr:thr:improvedChenUpper2}
Let $2 \leq u \leq d < k = d + u \leq n$ be integers with $(d + u)^2/u \leq n$. Set $\alpha = k \ln{\frac{\mathrm{e} n}{k}} + u \ln{\frac{\mathrm{e} k}{u}}$ and $\beta = 1 - 2/\alpha$. For any integer $z \geq 4/\beta^2 + 1$, there exists an $h \times n$ $(n, d, u; z]$-disjunct matrix with
\begin{align}
&h(n, d, u; z] \nonumber \\
=& \left\lfloor \frac{2}{\delta^2} \cdot \left( \frac{k}{u} \right)^u \left( \frac{k}{d} \right)^d \cdot \left( k \ln{\frac{\mathrm{e} n}{k}} + u \ln{\frac{\mathrm{e} k}{u}}  \right)  \right\rfloor + 1 \nonumber \\
=& O \left( \frac{1}{\delta^2} \cdot \left( \frac{k}{u} \right)^u \left( \frac{k}{d} \right)^d \cdot k \ln{\frac{n}{k}}  \right) \nonumber \\
<& t(n, d, u; z] = z \left( \frac{k}{u} \right)^u \left( \frac{k}{d} \right)^d \left[ 1 + k \left(1 + \ln{\left( \frac{n}{k} + 1 \right)} \right) \right], \nonumber
\end{align}
where $0 < \delta \leq \beta$.
\end{theorem}

\begin{proof}
By using the same construction and arguments in the proof in Theorem~\ref{thr:improvedChenUpper} until~\eqref{eqn:Exist3}, if we set 
\begin{align}
h &= h(n, d, u; z] \nonumber \\
&= \left\lfloor \frac{2}{\delta^2} \cdot \frac{1}{q} \cdot \left( k \ln{\frac{\mathrm{e} n}{k}} + u \ln{\frac{\mathrm{e} k}{u}}  \right) \right\rfloor + 1 \nonumber \\
&= \left\lfloor \frac{2}{\delta^2} \cdot \left( \frac{k}{u} \right)^u \left( \frac{k}{d} \right)^d \cdot \left( k \ln{\frac{\mathrm{e} n}{k}} + u \ln{\frac{\mathrm{e} k}{u}}  \right)  \right\rfloor + 1 \nonumber \\
&= O \left( \frac{1}{\delta^2} \cdot \left( \frac{k}{u} \right)^u \left( \frac{k}{d} \right)^d \cdot k \ln{\frac{n}{k}}  \right) \nonumber \\
&= O \left( \frac{1}{(1 - \delta) k \ln{\frac{n}{k}} } \cdot z \left( \frac{k}{u} \right)^u \left( \frac{k}{d} \right)^d \cdot k \ln{\frac{n}{k}}  \right) \label{eqn:z1} \\
&= O \left( \frac{1}{(1 - \delta) k \ln{\frac{n}{k}} } \right) \cdot t(n, d, u; z], \nonumber
\end{align}
then $g(p, h, u, d, n) < 1$, where $t(n, d, u; z]$ is defined in~\eqref{eqn:ChenUpper}; i.e., there exists an $(n, d, u; z]$-disjunct matrix of size $h \times n$. Equation~\eqref{eqn:z1} is obtained because 
\begin{align}
&& &\frac{2(1 - \delta)}{\delta^2} \cdot \left( k \ln{\frac{\mathrm{e} n}{k}} + u \ln{\frac{\mathrm{e} k}{u}}  \right) \nonumber \\
&&\leq z &= (1 - \delta) q h \label{eq:z21} \\
&& &= (1 - \delta)q \left( \left\lfloor \frac{2}{\delta^2} \cdot \frac{1}{q} \cdot \left( k \ln{\frac{\mathrm{e} n}{k}} + u \ln{\frac{\mathrm{e} k}{u}}  \right) \right\rfloor + 1 \right) \nonumber \\
&& &= \Theta \left( \frac{1 - \delta}{\delta^2} \cdot k \ln{\frac{n}{k}} \right) \nonumber \\ %\Theta \left( \frac{1 - \delta}{\delta^2} \left( k \ln{\frac{n}{k}} + u \ln{\frac{k}{u}}  \right) \right) =
&& &\leq \frac{2(1 - \delta)}{\delta^2} \cdot \left( k \ln{\frac{\mathrm{e} n}{k}} + u \ln{\frac{\mathrm{e} k}{u}}  \right) + 1. \label{eq:upperZ1}
\end{align}

We next prove that $h(n, d, u; z] < t(n, d, u; z]$ once $0 < \delta \leq 1 - \frac{2}{k \ln{\frac{\mathrm{e} n}{k}} + u \ln{\frac{\mathrm{e} k}{u}} }$. Indeed, we have
\begin{align}
&h(n, d, u; z] \nonumber \\
=& \left\lfloor \frac{2}{\delta^2} \cdot \frac{1}{q} \cdot \left( k \ln{\frac{\mathrm{e} n}{k}} + u \ln{\frac{\mathrm{e} k}{u}} \right) \right\rfloor + 1, \mbox{ where } \frac{1}{q} = \left( \frac{k}{u} \right)^u \left( \frac{k}{d} \right)^d \nonumber \\
\leq& \frac{2}{\delta^2} \cdot \frac{1}{q} \cdot \left( k \ln{\frac{\mathrm{e} n}{k}} + u \ln{\frac{\mathrm{e} k}{u}} \right) + 1 \nonumber \\
<& \frac{2}{\delta^2} \cdot \frac{1}{q} \cdot 2k \ln{\frac{n}{k}}. \label{eq:smaller:11}
\end{align}

This equation is attained because $k \ln{\frac{\mathrm{e} n}{k}} + u \ln{\frac{\mathrm{e} k}{u}} < 2k \ln{\frac{\mathrm{e} n}{k}}$ as $(d + u)^2/u \leq n$. On the other hand, we have
\begin{align}
&t(n, d, u; z] \nonumber \\
=& z \cdot \frac{1}{q} \cdot \left[ 1 + k \left(1 + \ln{\left( \frac{n}{k} + 1 \right)} \right) \right] \nonumber \\
>& z \cdot \frac{1}{q} \cdot k \ln{\frac{n}{k}} \nonumber \\
\geq& \frac{2(1 - \delta)}{\delta^2} \cdot \left( k \ln{\frac{\mathrm{e} n}{k}} + u \ln{\frac{\mathrm{e} k}{u}}  \right) \cdot \frac{1}{q} \cdot k \ln{\frac{n}{k}}, \label{eq:smaller:21}
\end{align}
which is derived from the condition in~\eqref{eq:z21}. Combining~\eqref{eq:smaller:21} and~\eqref{eq:smaller:11}, we always get $h(n, d, u; z] < t(n, d, u; z]$ if
\begin{align}
&& \frac{2}{\delta^2} \cdot \frac{1}{q} \cdot 2k \ln{\frac{n}{k}} \leq& \frac{2(1 - \delta)}{\delta^2} \cdot \left( k \ln{\frac{\mathrm{e} n}{k}} + u \ln{\frac{\mathrm{e} k}{u}} \right) \cdot \frac{1}{q} \cdot k \ln{\frac{n}{k}} \nonumber \\
&\Longleftrightarrow& \delta \leq& 1 - \frac{2}{k \ln{\frac{\mathrm{e} n}{k}} + u \ln{\frac{\mathrm{e} k}{u}} } = \beta. \nonumber
\end{align}

Since $0 < \delta \leq \beta = 1 - 2/\alpha$, the quantity $2(1 - \delta)/\delta^2 \cdot \alpha$ goes from $4/\beta^2$ to infinity. Moreover, from~\eqref{eq:z21} and~\eqref{eq:upperZ1}, we have $z \in \left[ \frac{2(1 - \delta)}{\delta^2} \cdot \alpha, \frac{2(1 - \delta)}{\delta^2} \cdot \alpha + 1  \right]$, where $\alpha = k \ln{\frac{\mathrm{e} n}{k}} + u \ln{\frac{\mathrm{e} k}{u}}$. Therefore, $z$ can range from $\lceil 4/\beta^2 \rceil$ to $+\infty$. In other words, for any integer $z \geq 4/\beta^2 + 1$, we can find a corresponding $\delta$ in the interval $(0, \beta]$ such that $z = (1 - \delta) qh$.
\end{proof}

\section{Improved non-adaptive algorithms for threshold group testing with a gap}
\label{main:sec:improvedEncDec}

Here we present a reduced exact upper bound on the number of tests and a reduced asymptotic bound on the decoding time for identifying defective items in a noisy setting on test outcomes compared with the state-of-the-art work of Chen and Fu~\cite{chen2009nonadaptive}.

\subsection{First proposed algorithm}
\label{main:sub:improvedEncDec:first}

By using the construction of an $(n, d - \ell, u; z]$-disjunct matrix described in Section~\ref{main:sec:improvedDisjunct}, we can reduce the number of tests for encoding and the decoding time for decoding in TGT with a gap. From Chen and Fu's work~\cite{chen2009nonadaptive}, if we use the $(n, d - \ell, u; z]$-disjunct matrix described in Theorem~\ref{thr:improvedChenUpper} as the input to Algorithm~\ref{alg:decodingThreshold}, the following theorem is derived:

\begin{theorem}
\label{thr:improved0}
Let $\ell, 0 < g, 2 \leq u = \ell + g + 1 \leq d < k = d - \ell + u \leq n$ be integers with $(d + u)^2/u \leq n$. Set $\alpha = k \ln{\frac{\mathrm{e} n}{k}} + u \ln{\frac{\mathrm{e} k}{u}}$. Let $z$ be a positive integer and $S$ be the defective set with $|S| \leq d$. For an $(n, d, \ell, u)$-TGT model with at most $e = \lfloor (z - 1)/2 \rfloor$ erroneous outcomes, there exists a non-adaptive algorithm that successfully identifies some set $S^\prime$ with $|S^\prime \setminus S| \leq g$ and $|S \setminus S^\prime| \leq g$ using no more than $h(n, d - \ell, u; z]$ tests, where $h(n, d - \ell, u; z]$ is defined in~\eqref{eqn:improvedChenUpper}. Moreover, the decoding complexity is 
\begin{align}
O \left( h(n, d - \ell, u; z] \times u \left( \binom{n}{u} + (d - u) \binom{n - u}{g + 1} \binom{d - 1}{g} \binom{d}{u} \right) \right). \label{eq:improved1}
\end{align}
\end{theorem}

\subsection{Second proposed algorithm}
\label{main:sub:improvedEncDec:second}

We can see that the complexity of the decoding algorithm in the theorem above remains relatively high due to the second operator in~\eqref{eq:improved1}. To reduce the decoding complexity, one can relax the conditions on $|S^\prime \setminus S|$ but keep the same condition on $|S \setminus S^\prime|$. In other words, we accept more false positives while keeping the same condition on the maximum number of false negatives.

The main idea is to reduce the redundancy of $u$-subsets created by the $(n, d - \ell, u; z]$-disjunct matrix in Algorithm~\ref{alg:decodingThreshold}. Since every $u$-subset $X^+ \subseteq S$ must be in $\sfF$, the total number of such $X^+$ is $\binom{|S|}{u}$. In fact, we need only up to $\zeta = \left\lfloor \frac{|S|}{u} \right\rfloor$ disjoint $u$-subsets $X^+$s in $\sfF$ to form $S$ if $|S|$ is divisible by $u$. Therefore, we can use a simple procedure, i.e., collect $\zeta$ disjoint $u$-subsets in $\sfF$, to form $S$. However, it is uncertain whether each $u$-subset we collected is truly a $u$-subset of $S$ because it may contain only $\ell + 1$ defective items. Moreover, $|S|$ may not be divisible by $u$. As a result, the set formed, says $S^\prime$, may not be identical to $S$. To remedy this drawback, we propose adding one more step: add the remaining defective items in $S \setminus S^\prime$ into $S^\prime$ until $S^\prime$ is not extendible.

The above strategy is formalized in the following theorem which is associated with Algorithm~\ref{alg:improvedAlgorithm}.

\begin{algorithm}
\caption{$\mathrm{Decoding}_2(\bY, \cM)$: Decoding procedure for non-adaptive $(n, d, \ell, u)$-TGT with up to $e$ erroneous outcomes.}
\label{alg:improvedAlgorithm}
\textbf{Input:} Outcome vector $\bY$, an $(n, d - \ell, u; z = 2e + 1]$-disjunct matrix $\cM$.\\
\textbf{Output:} Set of defective items $S^\prime$ s.t. $|S^\prime \setminus S| \leq \left( \left\lfloor \frac{|S|}{\ell + 1} \right\rfloor + u - 1 \right) g$ and $|S \setminus S^\prime| \leq g$.

\begin{algorithmic}[1]
\State Construct a family $\sfF$ such that a $u$-subset $X \subseteq [n]$ is an edge in $\sfF$ iff $t_0^{\cM}(X) \leq e$, where $t_0^{\cM}(X)$ is the number of negative pools in which all columns in $X$ appear when using $\cM$ as a measurement matrix. \label{alg:improvedAlgorithm:findFamily}
\State We first want to establish increasing vertex-sets $S_i$'s, $|S_1| < |S_2| \ldots < |S_{r}|$, such that $S_{i + 1}$ contains exactly $u$ items more than $S_i$. As an initial $S_1$, we select all $u$ vertices of an arbitrary edge. To find $S_{i + 1}$ for $i \geq 1$, we check all possible cases to attain some $u$-subset $A_i \in \sfF \setminus \{A_1, \ldots, A_{i - 1} \}$ such that $|S_i \cup A_i| = |S_i| + u$. If $A_i$ exists, then set $S_{i + 1} = S_i \cup A_i$. This process is continued until $S_r$ is not extendible. \label{alg:improvedAlgorithm:getDefectiveSetX1}
\State We then want to establish increasing vertex-sets $S_i$'s, $|S_{r + 1}| < |S_{r + 2}| \ldots < |S_m|$, such that $S_{i + 1}$ contains at least one defective item more than $S_i$. To find $S_{i + 1}$ for $i \geq r$, we check all possible cases to attain some $u$-subset $A_i \in \sfF \setminus \{A_1, \ldots, A_{i - 1} \}$ such that $|S_i \cup A_i| \geq |S_i| + g + 1$. If $A_i$ exists, then set $S_{i + 1} = S_i \cup A_i$. This process is continued until $S_m$ is not extendible. Output set $S^\prime = S_m$. \label{alg:improvedAlgorithm:getDefectiveSetX2}
\end{algorithmic}
\end{algorithm}

\begin{theorem}
\label{thr:mainImprovement}
Let $\ell, 0 < g, 2 \leq u = \ell + g + 1 < d < k = d - \ell + u \leq n$ be integers with $\mathrm{e}^2 (d + u)^2/u \leq n$. Set $\alpha = k \ln{\frac{\mathrm{e} n}{k}} + u \ln{\frac{\mathrm{e} k}{u}}$. Let $z$ be a positive integer and $S$ be the defective set with $|S| \leq d$. For an $(n, d, \ell, u)$-TGT model with at most $e = \lfloor (z - 1)/2 \rfloor$ erroneous outcomes, there exists a non-adaptive algorithm that successfully identifies some set $S^\prime$ with $|S^\prime \setminus S| \leq \left( \left\lfloor \frac{|S|}{\ell + 1} \right\rfloor + u - 1 \right) g \leq \left( \frac{d}{\ell + 1} + u - 1 \right) g$ and $|S \setminus S^\prime| \leq g$ using no more than $h(n, d - \ell, u; z]$ tests, where $h(n, d - \ell, u; z]$ is defined in~\eqref{eqn:improvedChenUpper}. Moreover, the decoding complexity is
\begin{equation}
O \left( h(n, d - \ell, u; z] \cdot u\binom{n}{u} \right). \nonumber
\end{equation}
\end{theorem}

The proof of this theorem is divided into two parts: correctness and decoding complexity. However, we first present visualizations that convey the essence of Algorithm~\ref{alg:improvedAlgorithm}.

\subsubsection{Visualization}

Steps~\ref{alg:improvedAlgorithm:getDefectiveSetX1} and~\ref{alg:improvedAlgorithm:getDefectiveSetX2} of Algorithm~\ref{alg:improvedAlgorithm} are depicted in Fig.~\ref{fig:Alg.2} for $n = 49, g = 7, u = 10, l = u - g - 1 = 2, d = 17$, and $|S| = 17$. There are many $u$-subsets belonging to $\sfF$, but we depict only five of them here. Step~\ref{alg:improvedAlgorithm:getDefectiveSetX1} proceeds as shown in the upper five images as follows. Subset $S_1$, containing $10$ defectives, selected as the initial subset. Scanning every subset of $\sfF$ reveals that $A_1$ is a subset such that $|S_1 \cup A_1| = |S_1| + |A_1| = 20$. Set $S_2 = S_1 \cup A_1$. The process continues until $S_r$ consists of $13$ defectives and $7$ negatives. Since there are no $u$-subsets $A_{r}$'s in $\sfF \setminus \{A_1, \ldots, A_{r - 1} \}$ such that $|S_r \cup A_r| = |S_r| + |A_r|$, $S_r$ is not extendible.

Step~\ref{alg:improvedAlgorithm:getDefectiveSetX2} proceeds as shown in the lower four images. Starting with subset $S_r$, we try to find a $u$-subset $A_{r}$'s in $\sfF \setminus \{A_1, \ldots, A_{r - 1} \}$ such that $|S_r \cup A_r| = |S_r| + g + 1$. If $A_r$ exists, a new subset $S_{r + 1} = S_r \cup A_r$ is attained. This process is repeated until there are no $u$-subsets $A_{m}$'s in $\sfF \setminus \{A_1, \ldots, A_{m - 1} \}$ such that $|S_m \cup A_m| \geq |S_m| + g + 1$. In other words, $S_m$ is not extendible. The algorithm terminates and $S^\prime = S_m$ is attained.

\begin{figure*}[t]
\begin{subfigure}{\textwidth}
\centering
  \includegraphics[scale=0.5]{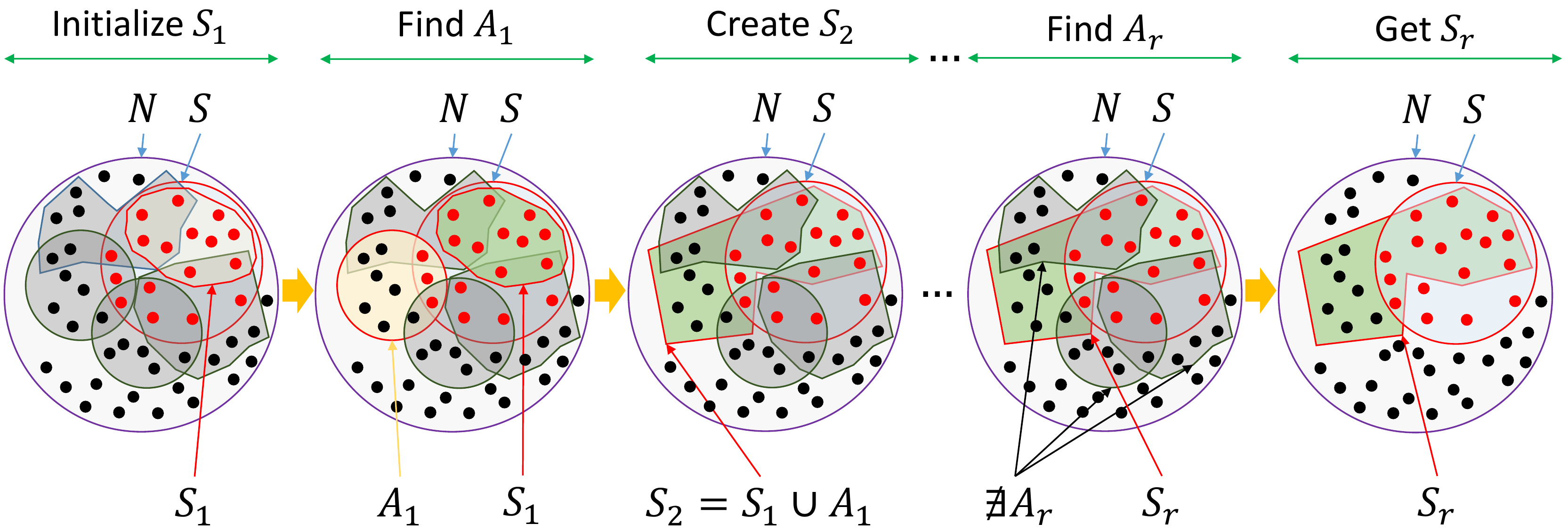}
  \caption{Step~\ref{alg:improvedAlgorithm:getDefectiveSetX1}}
  \label{fig:sub-first}
\end{subfigure}\\
\begin{subfigure}{\textwidth}
  \centering
  \includegraphics[scale=0.5]{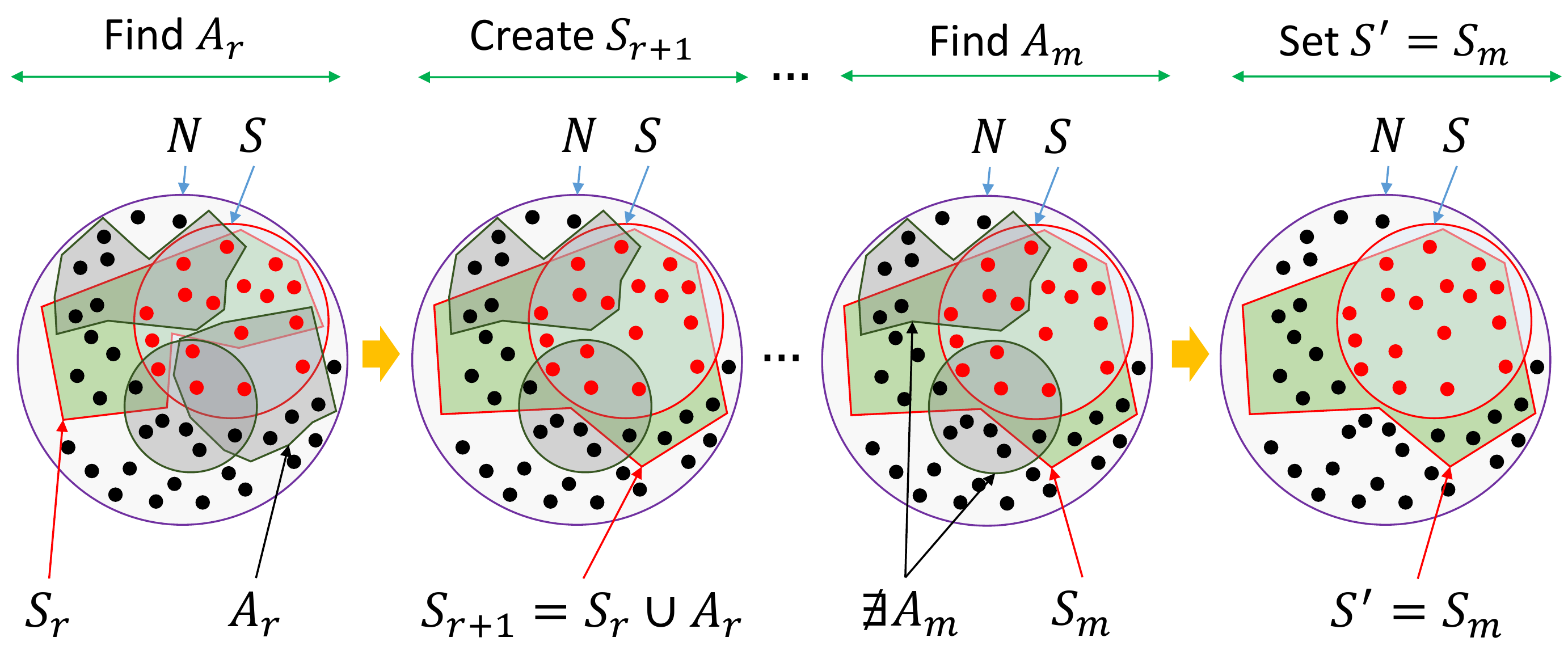}
  \caption{Step~\ref{alg:improvedAlgorithm:getDefectiveSetX2}}
  \label{fig:sub-second}
\end{subfigure}

\caption{Illustration of finding an approximate defective set $S^\prime$ of defective set $S$ such that $|S^\prime \setminus S| \leq \left( \left\lfloor \frac{|S|}{\ell + 1} \right\rfloor + u - 1 \right) g$ and $|S \setminus S^\prime| \leq g$ with $g = 7, u = 10$, and $\ell = u - g - 1 = 2$ for Algorithm~\ref{alg:improvedAlgorithm}.}
\label{fig:Alg.2}
\end{figure*}

\subsubsection{Correctness}
\label{sub:main:improvedEncDec:correct}

To prove the correctness of Algorithm~\ref{alg:improvedAlgorithm}, we first prove that after Step~\ref{alg:improvedAlgorithm:findFamily}, for every $u$-subset $X \in \sfF$, $X$ contains no more than $g$ items not in $S$, i.e., $|X \cap S| \geq \ell + 1$. Moreover, every $u$-subset $X^+ \subseteq S$ is in $\sfF$. Since the proof is identical to the proof of Lemma 4.1 in~\cite{chen2009nonadaptive}, we omit it here.

Since every $u$-subset $X^+ \subseteq S$ must be in $\sfF$, there exists $\left\lfloor \frac{|S|}{u} \right\rfloor \leq \zeta \leq \left\lceil \frac{|S|}{u} \right\rceil$ disjoint $u$-subsets $X_1^+, \ldots, X_\zeta^+$ in $\sfF$ such that $|S \setminus \cup_{j = 1}^\zeta X_j^+| \leq u - 1$.

In Step~\ref{alg:improvedAlgorithm:getDefectiveSetX1}, if another disjoint $u$-subset $A_1 \in \sfF$ ($|S_1 \cap A_1| = 0$) is found, set $S_2 = S_1 \cup A_1$. In general, to find $S_{i + 1}$ for $i \geq 1$, all possible cases are checked to attain some $u$-subset $A_i \in \sfF \setminus \{A_1, \ldots, A_{i - 1} \}$ such that $|S_i \cup A_i| = |S_i| + |A_i| = |S_i| + u$. If $A_i$ exists, i.e., $S_i$ is extendible, set $S_{i + 1} = S_i \cup A_i$. On the other hand, if $S_r$ is not extendible ($A_i$ does not exist), we can infer that $|S \setminus S_r| \leq u - 1$. We assume that $|S \setminus S_r| \geq u$. Select $A_r \subseteq S \setminus S_r$ with $|A_r| = u$. Since $A_r \subseteq S \setminus S_r \in \sfF$ and $|A_r \cap S_r| = 0$, we get $|S_r \cup A_r| = |S_r| + |A_r| = |S_r| + u$; i.e., $S_r$ is extendible. This contradicts the assumption.

There is a special case that if $|S \setminus S_r| \leq \ell$, this process stops. If $|S \setminus S_r| \leq \ell$, for any $A_r \in \sfF \setminus \{A_1, \ldots, A_{r - 1} \}$, we have $|A_r \cap S_r| \geq 1$ because $|A_r \cap S| \geq \ell + 1$. Therefore, there does not exist $A_r \in \sfF$ such that $|S_r \cup A_r| = |S_r| + u$ because $|S_r \cup A_r| = |S_r| + |A_r| - |S_r \cap A_r| \leq |S_r| + u - 1 < |S_r| + u$.

Because each $A_i$ can contain exactly $\ell + 1$ defectives in the worst case, Step~\ref{alg:improvedAlgorithm:getDefectiveSetX1} can run up to $\left\lfloor \frac{|S|}{\ell + 1} \right\rfloor$ times. 
 
We now consider Step~\ref{alg:improvedAlgorithm:getDefectiveSetX2}. Subset $S_m$ is not extendible iff there does not exist a $u$-subset $A_m \in \sfF \setminus \{A_1, \ldots, A_{m - 1} \}$ such that $|S_m \cup A_m| \geq |S_m| + g + 1$. We then must have $|S \setminus S_m| \leq g$. Indeed, let us assume that $|S \setminus S_m| \geq g + 1$. Select $C \subseteq S \setminus S_m$ with $|S| = g + 1$ and $D \subseteq S \setminus C$ with $|D| = \ell$. Such a pair $C, D$ always exists because $|S| \geq u = g + 1 + \ell$. Set $A_m = C \cup D$. Therefore, $|S_m \cup A_m| \geq |S_m| + g + 1$ and $A_m \in \sfF$. Hence, $S_m$ is extendible, which contradicts the assumption that $S_m$ is not extendible.

We have $|S \setminus S_r| \leq u - 1$ after running Step~\ref{alg:improvedAlgorithm:getDefectiveSetX1}. It follows that Step~\ref{alg:improvedAlgorithm:getDefectiveSetX2} runs at most $(u - 1)$ times, i.e., $m - r \leq u - 1$, because $S_i$ adds at least one defective for each iteration of Step~\ref{alg:improvedAlgorithm:getDefectiveSetX2}.

In summary, Steps~\ref{alg:improvedAlgorithm:getDefectiveSetX1} and~\ref{alg:improvedAlgorithm:getDefectiveSetX2} run up to $\left\lfloor \frac{|S|}{\ell + 1} \right\rfloor$ and $(u - 1)$ times, respectively. Because the subset considered at each iteration adds a $u$-subset having at least $\ell + 1$ defectives and up to $g$ negatives, we have $|S^\prime \setminus S| \leq \left( \left\lfloor \frac{|S|}{\ell + 1} \right\rfloor + u - 1 \right) g$ and $|S \setminus S^\prime| \leq g$ when the algorithm terminates.

\subsubsection{Decoding complexity}

Step~\ref{alg:improvedAlgorithm:findFamily} takes $h(n, d - \ell, u; z] \cdot u\binom{n}{u}$ time. Since every $u$-subset in $\sfF$ has at least $\ell + 1$ defectives and up to $g = u - \ell + 1$ negatives, the maximum cardinality of $\sfF$ is:
\begin{equation}
f = \sum_{i = \ell + 1}^{u} \binom{|S|}{i} \binom{n - |S|}{u - i} < \sum_{i = 0}^{u} \binom{|S|}{i} \binom{n - |S|}{u - i} = \binom{n}{u}. \nonumber
\end{equation}

Because $|S^\prime \setminus S| \leq \left( \left\lfloor \frac{|S|}{\ell + 1} \right\rfloor + u - 1 \right) g$ and $|S \setminus S^\prime| \leq g$, we have $|S^\prime| \leq \left( \left\lfloor \frac{|S|}{\ell + 1} \right\rfloor + u - 1 \right)g + d$. Since we scan the family $\sfF$ up to $\left\lfloor \frac{|S|}{\ell + 1} \right\rfloor + (u - 1)$ times in both Steps~\ref{alg:improvedAlgorithm:getDefectiveSetX1} and~\ref{alg:improvedAlgorithm:getDefectiveSetX2}, $|\sfF| \leq f$, and $|S_i| \leq |S^\prime| \leq \left( \left\lfloor \frac{|S|}{\ell + 1} \right\rfloor + u - 1 \right) g + d$, the complexity of Algorithm~\ref{alg:improvedAlgorithm} is:
\begin{align}
&h(n, d - \ell, u; z] \cdot u\binom{n}{u} + \left( \left\lfloor \frac{|S|}{\ell + 1} \right\rfloor + u - 1 \right) \left( \left( \left\lfloor \frac{|S|}{\ell + 1} \right\rfloor + u - 1 \right) g + d \right) \times u f \nonumber \\
=& h(n, d - \ell, u; z] \cdot u\binom{n}{u} + us (gs + d) f, \label{eq:counting}
\end{align}
where $s = \left\lfloor \frac{|S|}{\ell + 1} \right\rfloor + (u - 1) \leq d + u$ and $k = d - \ell + u = d + g + 1$.

We have
\begin{align}
us (gs + d) f &\leq u ( d + u) (g(d + u) + d) \binom{n}{u} \nonumber \\
&< (d + u)^2 (g + 1) \cdot u \binom{n}{u}, \label{eq:cmp1}
\end{align}
and
\begin{align}
&h(n, d - \ell, u; z] \cdot u\binom{n}{u} \nonumber \\
\geq& \left( 1 + \frac{d - \ell}{u} \right)^u \left( 1 + \frac{u}{d - \ell} \right)^{d - \ell} (d + g + 1) \ln{\frac{n}{k}} \cdot u\binom{n}{u} \nonumber \\
\geq& 4(g + 1) \left( 1 + \frac{d - \ell}{u} \right)^u \left( 1 + \frac{u}{d - \ell} \right)^{d - \ell} \cdot u\binom{n}{u}, \label{eq:cmp2}
\end{align}
because $h(n, d - \ell, u; z] \geq \left( 1 + \frac{d - \ell}{u} \right)^u \left( 1 + \frac{u}{d - \ell} \right)^{d - \ell} (d + g + 1) \ln{\frac{n}{k}}$ as in Theorem~\ref{thr:improvedChenUpper}, $d \geq u \geq g + 1$ and $\ln{\frac{n}{k}} \geq 2$ ($n \geq \mathrm{e}^2 (d + u)^2/u > \mathrm{e}^2 (d - \ell + u)$). We next consider the following inequality:
\begin{align}
&& (d + u)^2 (g + 1) \cdot u \binom{n}{u} &\leq 4(g + 1) \left( 1 + \frac{d - \ell}{u} \right)^u \nonumber \times \left( 1 + \frac{u}{d - \ell} \right)^{d - \ell} \cdot u\binom{n}{u} & & \label{eq:cmp3} \\
&\Longleftrightarrow& d + u &\leq 2 \left( 1 + \frac{d - \ell}{u} \right)^{u/2} \left( 1 + \frac{u}{d - \ell} \right)^{(d - \ell)/2} \nonumber
\end{align}

For this inequality to hold, by using Bernoulli's inequality, it suffices that
\begin{align}
&& d + u &\leq 2 \left(  1 + \frac{d - \ell}{u} \times \frac{u}{2} \right) \left( 1 + \frac{u}{d - \ell} \times \frac{d - \ell}{2} \right) \nonumber \\
&& &\leq 2 \left( 1 + \frac{d - \ell}{u} \right)^{u/2} \left( 1 + \frac{u}{d - \ell} \right)^{(d - \ell)/2} \nonumber \\
&\Longleftrightarrow& d + u &\leq \frac{(d - \ell + 2)(u + 2)}{2} \nonumber \\
&& &\leq \frac{du}{2} + (d + u) + 2 - \frac{\ell (u + 2)}{2} \nonumber \\
&\Longleftrightarrow& \ell (u + 2) &\leq du + 4. \nonumber
\end{align}

The last inequality always holds because $\ell ( u + 2) \leq (u - 1) (u + 2) < u(u + 1) + 4 \leq du + 4$ for $d \geq u + 1$. Combining~\eqref{eq:cmp1},~\eqref{eq:cmp2}, and~\eqref{eq:cmp3}, we get
\begin{equation}
us (gs + d) f \leq h(n, d - \ell, u; z] \cdot u\binom{n}{u}, \nonumber
\end{equation}
for any $d \geq u + 1$ and $n \geq \mathrm{e}^2 (d + u)^2/u > \mathrm{e}^2 (d - \ell + u)$. Therefore, the decoding complexity of Algorithm~\ref{alg:improvedAlgorithm} is up to

\begin{equation}
h(n, d - \ell, u; z] \cdot 2u\binom{n}{u}. \nonumber
\end{equation}

\subsubsection{Example for Algorithm~\ref{alg:improvedAlgorithm}}
\label{subsub:improvedEncDec:eg2}

We demonstrate Algorithm~\ref{alg:improvedAlgorithm} by the same settings used to demonstrate Algorithm~\ref{alg:decodingThreshold} (Section~\ref{sub:review:eg}): $n = 6, d = 4, \ell = 0, u = 2, g = u - \ell - 1 = 1, z = 1, e = 0$, and $S = \{1, 2, 4, 5 \}$. Input vector $\bX$, $(n = 6, d - \ell = 4, u = 2; z = 1]$-disjunct matrix $\cM$, and outcome vector $\bY$ are as in~\eqref{eqn:exChenFu}. Note that the condition $\mathrm{e}^2 (d + u)^2/u \leq n$ does not hold though Algorithm~\ref{alg:improvedAlgorithm} still works well for this example.

Algorithm~\ref{alg:improvedAlgorithm} proceeds as follows. In Step~\ref{alg:improvedAlgorithm:findFamily}, a family $\sfF$ of $2$-subsets $X \subseteq [n]$ is constructed such that $t_0^{\cM}(X) \leq e = 0$. From~\eqref{eqn:exChenFu}, we get $\sfF = \{ \{1, 2\}, \{1, 4\}, \{1, 5\}, \{2, 3\}, \{2, 4\}, \{2, 5\}, \{3, 5\}, \{4, 5\}, \{5, 6\} \}$.

In Step~\ref{alg:improvedAlgorithm:getDefectiveSetX1}, an initial $2$-subset $S_1$ belonging to $\sfF$ is arbitrarily chosen. Without loss of generality, set $S_1 = \{1, 2 \}$. The next phase is to find a $2$-subset $A_1 \in \sfF$ such that $|S_1 \cup A_1| = |S_1| + u = 2 + 2 = 4$. By exhausted searching in $\sfF$, one of candidates is $\{3, 5 \}$. Set $A_1 = \{3, 5 \}$ and $S_2 = S_1 \cup A_1 = \{1, 2, 3, 5 \}$. Because there does not exist any $A_2 \in \sfF \setminus \{A_1 \}$ such that $|S_2 \cup A_2| = |S_2| + 2$, Step~\ref{alg:improvedAlgorithm:getDefectiveSetX1} stops here because $S_2$ is not extendible.

Since set $S_2$ returned by Step~\ref{alg:improvedAlgorithm:getDefectiveSetX1} may not contain some of the defective items when $|S \setminus S_2| > g$, Step~\ref{alg:improvedAlgorithm:getDefectiveSetX2} exhaustively searches for them in order to produce the final approximate set $S^\prime$, where $|S \setminus S^\prime| \leq g$. It searches for a $2$-subset $A_2$ in $\sfF \setminus \{A_1 \}$ such that $|S_2 \cup A_2| = |S_2| + g + 1 = 3 + 1 + 1 = 5$. Luckily, such an $A_2$ does not exist, $S^\prime = S_2 = \{1, 2, 3, 5 \}$ is output.

Note that the approximate defective set $S^\prime$ is not identical to $S$ as it is in Section~\ref{sub:review:eg}. However, set $S^\prime$ is \textit{indistinguishable} from $S$ because $|S \setminus S^\prime| = 1 = g \leq g$ and $|S^\prime \setminus S| = 1 = g \leq g$~\cite{damaschke2006threshold}.

\subsection{Third proposed algorithm}
\label{sub:variant}

Our main idea here is to combine Algorithms~\ref{alg:decodingThreshold} and~\ref{alg:improvedAlgorithm}. It is obvious that $|S^\prime \setminus S| \leq \left( \left\lfloor \frac{|S|}{\ell + 1} \right\rfloor + u - 1 \right) g$ in Theorem~\ref{thr:mainImprovement}, which is worse than the condition $|S^\prime \setminus S| \leq g$ in Theorem~\ref{thr:improved0}. Theorem~\ref{thr:mainImprovement} can be improved to achieve the conditions $|S^\prime \setminus S| \leq g$ and $|S \setminus S^\prime| \leq 2g$ by using the outcome of Algorithm~\ref{alg:improvedAlgorithm} as the input of Algorithm~\ref{alg:decodingThreshold}. An extension of Algorithm~\ref{alg:improvedAlgorithm} is described in Algorithm~\ref{alg:improvedAlgorithm2}. The decoding complexity of the improved algorithm is higher than that in Theorem~\ref{thr:mainImprovement} but lower than that in Theorem~\ref{thr:improved0}. The conditions on $|S^\prime \setminus S|$ and $|S \setminus S^\prime|$, i.e., the number of false positives and the number of false negatives, are respectively looser than and equal to the corresponding ones in Theorem~\ref{thr:mainImprovement}. On the other hand, the conditions on $|S^\prime \setminus S|$ and $|S \setminus S^\prime|$ are equal to and tighter than the corresponding ones in Theorem~\ref{thr:improved0}. These comparisons are summarized in Table~\ref{buice.t1}.

\begin{algorithm}
\caption{$\mathrm{Decoding}_3(\bY, \cM)$: Decoding procedure for non-adaptive $(n, d, \ell, u)$-TGT with up to $e$ erroneous outcomes.}
\label{alg:improvedAlgorithm2}
\textbf{Input:} Outcome vector $\bY$, a $(d - \ell, u; z = 2e + 1]$-disjunct matrix $\cM$.\\
\textbf{Output:} Set of defective items $S^\prime$ s.t. $|S^\prime \setminus S| \leq g$ and $|S \setminus S^\prime| \leq 2g$.

\begin{algorithmic}[1]
\State Set $V = \mathrm{Decoding}_2(\bY, \cM)$. \label{alg:improvedAlgorithm2:findDefectiveSet}
\State Construct hypergraph $\mathbb{H} = (V, \sfF)$ where a $u$-subset $X \subseteq V$ is an edge in $\sfF$ iff $t_0^{\cM}(X) \leq e$, where $t_0^{\cM}(X)$ is the number of negative pools in which all columns in $X$ appear when using $\cM$ as a measurement matrix. \label{alg:improvedAlgorithm2:findFamily}
\State We want to establish increasing vertex-sets $S_i$'s, $|S_1| < |S_2| \ldots < |S_m|$ such that hypergraph $\mathbb{H}$ is $u$-complete with respect to each $S_i$. As an initial $S_1$, we can select all $u$ vertices of an arbitrary edge. To find $S_{i + 1}$ for $i \geq 1$, we check all possible cases to attain some $(g + 1)$-subset $A_i$ in $V(\mathbb{H}) \setminus S_i$ and a $g$-subset $B_i$ in $S_i$ such that $\mathbb{H}$ is $u$-complete with respect to $(S_i \cup A) \setminus B$. If such a pair $A_i, B_i$ exists, set $S_{i + 1} = (S_i \cup A_i) \setminus B_i$. This process is continued until either $S_m$ is not extendable or $|S_i| \geq d$. Output the set $S^\prime = S_m$. \label{alg:improvedAlgorithm2:getDefectiveSet}
\end{algorithmic}
\end{algorithm}

The set $S^\prime$ attained from Algorithm~\ref{alg:improvedAlgorithm2} satisfies two properties: $|S \setminus S^\prime| \leq 2g$ and $|S^\prime \setminus S| \leq g$. This can be interpreted to mean that the number of defective items in $S^\prime$, i.e., $|S^\prime \cap S| \geq |S| - 2g$, is at least $|S| - 2g$ . We summarize this result as follows.

\begin{theorem}
\label{thr:mainImprovement2}
Let $\ell, 0 < g, 2 \leq u = \ell + g + 1 < d < k = d - \ell + u \leq n$ be integers with $\mathrm{e}^2 (d + u)^2/u \leq n$. Let $z$ be a positive integer and $S$ be the defective set with $|S| \leq d$. Set $w = \left( \left\lfloor \frac{|S|}{\ell + 1} \right\rfloor + u - 1 \right) g$ and $w + d \leq n$. For an $(n, d, \ell, u)$-TGT model with at most $e = \lfloor (z - 1)/2 \rfloor$ erroneous outcomes, there exists a non-adaptive algorithm that successfully identifies some set $S^\prime$ with $|S^\prime \setminus S| \leq g$ and $|S \setminus S^\prime| \leq 2g$ using no more than $h(n, d - \ell, u; z]$ tests, where $h(n, d - \ell, u; z]$ is defined in~\eqref{eqn:improvedChenUpper}. Moreover, the decoding complexity is 
\begin{align}
O \left( h(n, d - \ell, u; z] \cdot u \cdot \left( \binom{n}{u} + (d - u) \binom{w + d - u}{g + 1} \binom{d - 1}{g} \binom{d}{u} \right) \right). \label{eq:improved2}
\end{align}
\end{theorem}

As with the previous one, the proof is divided into two parts: correctness and decoding complexity.

\subsubsection{Correctness}

From Theorem~\ref{thr:mainImprovement}, we get $|V \setminus S| \leq \left( \left\lfloor \frac{|S|}{\ell + 1} \right\rfloor + u - 1 \right) g$ and $|S \setminus V| \leq g$. Set $P = V \cap S$. We always have $|P| \geq |S| - g$ because $|S \setminus V| \leq g$.

Using the same argument as in the first paragraph of Section~\ref{sub:main:improvedEncDec:correct}, for any $u$-subset $X \in \sfF$, we get $|X \cap S| \geq \ell + 1$ and every $u$-subset $X^+ \subseteq P$ must be in $\sfF$. Because $V(\mathbb{H})$ is $u$-complete with respect to $S^\prime = S_m$, we attain $|S^\prime \setminus S| \leq g$.

We now show that $|S \setminus S^\prime| \leq 2g$ once $S^\prime = S_m$ is not extendable or $|S_m| \geq d$. Consider the case $|S^\prime| \geq d$. Since $|S \setminus S^\prime| \leq g$, we get $|S^\prime \cap S| \geq d - g$. This indicates that $|S \setminus S^\prime| \leq g \leq 2g$ because $|S| \leq d$.

It is now adequate to show that if $S^\prime$ is not extendable, then $|S \setminus S^\prime| \leq 2g$. To prove this property, it suffices to prove $|P \setminus S^\prime| \leq g$. The property is then straightforwardly attained because $P \subseteq S$ and $|P| \geq |S| - g$. Assume for the sake of contradiction that $|P \setminus S^\prime| > g$. Set $A_m \subseteq P \setminus S^\prime$ and $|A_m| = g + 1$, and let $B_m$ be any subset with $S^\prime \setminus P \subseteq B_m \subset S^\prime$ and $|B_m| = g$. Subset $B_m$ always exists because $|S^\prime \setminus S| \leq g$ and the initial $S^\prime$ has $u > g$ elements. Therefore, $(S^\prime \cup A_m) \setminus B_m$ is contained in $P$. It follows that $\mathbb{H}$ is $u$-complete with respect to $(S^\prime \cup A_m) \setminus B_m$. This contradicts the assumption that $S^\prime$ is not extendable.

In summary, $|S \setminus S^\prime| \leq 2g$ and $|S^\prime \setminus S| \leq g$ are always attained after running Algorithm~\ref{alg:improvedAlgorithm2}.

\subsubsection{Complexity}

From Theorem~\ref{thr:mainImprovement}, the complexity of Step~\ref{alg:improvedAlgorithm2:findDefectiveSet} is $h(n, d - \ell, u; z] \cdot u\binom{n}{u}$.

Because $|V| \leq \left( \left\lfloor \frac{|S|}{\ell + 1} \right\rfloor + u - 1 \right) g + d = w + d$, the complexity of Step~\ref{alg:improvedAlgorithm2:findFamily} is $u h(n, d - \ell, u; z] \times \binom{|V|}{u} \leq u h(n, d - \ell, u; z] \times \binom{w + d}{u}$.

We can verify whether ``$\mathbb{H}$ is $u$-complete with respect to $(S_i \cup A_i) \setminus B_i$'' if $t_0^{\cM}(Z) \leq e$ for every $u$-subset $Z \subseteq V$. Using an argument similar to the one described in the second paragraph of Appendix, we get that the complexity of Step~\ref{alg:improvedAlgorithm2:getDefectiveSet} is $(d - u) \binom{w + d - u}{g + 1} \binom{d - 1}{g} \binom{d}{u} \times u h(n, d - \ell, u; z]$. The total complexity of Algorithm~\ref{alg:improvedAlgorithm2} is then at most
\begin{align}
& h(n, d - \ell, u; z] \cdot u\binom{n}{u} + u h(n, d - \ell, u; z] \times \binom{w + d}{u} \nonumber \\
&+ (d - u) \binom{w + d - u}{g + 1} \binom{d - 1}{g} \binom{d}{u} \times u h(n, d - \ell, u; z] \nonumber \\
&= h(n, d - \ell, u; z] \times u \left( \binom{n}{u} + (d - u) \binom{w + d - u}{g + 1} \binom{d - 1}{g} \binom{d}{u} \right) \label{eq:finalComplexity} \\
&= h(n, d - \ell, u; z] \times u \left( \binom{n}{u} + (d - u) \binom{ \left( \left\lfloor \frac{|S|}{\ell + 1} \right\rfloor + u - 1 \right) g + d - u}{g + 1} \binom{d - 1}{g} \binom{d}{u} \right). \nonumber
\end{align}

Equation~\eqref{eq:finalComplexity} is attained if we suppose that $w + d = \left( \left\lfloor \frac{|S|}{\ell + 1} \right\rfloor + u - 1 \right) g + d \leq u \left( \frac{d}{\ell + 1} + u - 1 \right) + d \leq n$. This condition is practical because $n$ is much larger than $d$.

\subsubsection{Example for Algorithm~\ref{alg:improvedAlgorithm2}}
\label{subsub:improvedEncDec:eg3}

We demonstrate Algorithm~\ref{alg:improvedAlgorithm2} using the same settings as before: $n = 6, d = 4, \ell = 0, u = 2, g = u - \ell - 1 = 1, z = 1, e = 0$ and $S = \{1, 2, 4, 5 \}$. Input vector $\bX$, $(n = 6, d - \ell = 4, u = 2; z = 1]$-disjunct matrix $\cM$, and outcome vector $\bY$ are as in~\eqref{eqn:exChenFu}. Note that the conditions $\mathrm{e}^2 (d + u)^2/u \leq n$ and $w + d = \left( \left\lfloor \frac{|S|}{\ell + 1} \right\rfloor + u - 1 \right)g + d \leq n$ do not hold though Algorithm~\ref{alg:improvedAlgorithm2} still works well for this example.

Algorithm~\ref{alg:improvedAlgorithm2} proceeds as follows. In Step~\ref{alg:improvedAlgorithm2:findDefectiveSet}, the set of vertices $V = \{1, 2, 3, 5 \}$ is first obtained as described in Section~\ref{subsub:improvedEncDec:eg2}. Our task now is to construct a hypergraph $\mathbb{H} = (V, \sfF)$ using that set. Note that the original set of vertices, $[n] = [6] = \{1, 2, 3, 4, 5, 6 \}$, is here reduced to $V$. Using the same procedure described in Section~\ref{sub:review:eg}, Step~\ref{alg:improvedAlgorithm2:findFamily} searches for all 2-subsets $X \subseteq V$ in order to form a set of edges $\sfF$ such that $t^{\cM}_0(X) \leq e = 0$. From~\eqref{eqn:exChenFu}, we get $\sfF = \{ \{1, 2\}, \{1, 5\}, \{2, 3\}, \{2, 5\}, \{3, 5\} \}$.

Step~\ref{alg:improvedAlgorithm2:getDefectiveSet} starts with an initial $2$-subset $S_1 = \{1, 2\}$ and checks all possible cases to obtain some 2-subset $A_1$ in $V \setminus S_1 = \{ 3, 5\}$, which is $\{3, 5\}$, and a 1-subset $B_1$ in $S_1$, which is some element of $\{ \{1 \}, \{2 \} \}$, such that $\mathbb{H}$ is $2$-complete with respect to $(S_1 \cup A_1) \setminus B_1$. Since $A_1 = \{3, 5\}$ and $B_1 = \{1 \}$ ensure that the condition holds, set $S_2 = (S_1 \cup A_1) \setminus B_1 = \{2, 3, 5 \}$. Next, a 2-subset $A_2 \subseteq V \setminus S_2 = \{1 \}$ and a 1-subset $B_2 \subseteq S_2$ are chosen. Since $|V \setminus S_2| = 1 < 2 = u$, there does not exist such an $A_2$; i.e., $S_2$ is not extendible. Step~\ref{alg:improvedAlgorithm2:getDefectiveSet} thus stops and outputs $S^\prime = S_2 = \{2, 3, 5 \}$.

In this example, approximate defective set $S^\prime$ satisfies the two conditions $|S^\prime \setminus S| \leq g$ and $|S \setminus S^\prime| \leq 2g$ in Theorem~\ref{thr:mainImprovement2} because $|S^\prime \setminus S| = |\{ 3 \}| = 1 = g \leq g$ and $|S \setminus S^\prime| = |\{1, 4 \}| = 2 = 2g \leq 2g$.

\section{Simulation}
\label{sec:simul}

% Cheraghchi	z = p/(1 - p)^2 2d^2 \log(n/d)				t = d^{g + 2} \log(n / d) 2(8u)^u/ (1 - p)^2
% Chen et al.	z											t = z (k/u)^u (k/(d - \ell))^{d - \ell} [1 + k (1 + \log(n/k + 1)) ]
% Proposed		z = (1 - \delta)/\delta^2 \cdot k \log(n/k)	t = 2/\delta^2 (k/u)^u (k/(d - \ell))^{d - \ell} (k \log(en/k) + u \log(en/ u))

We visualized (upper bounds on) the number of tests for threshold group testing with a gap using five parameters $n, d, u, \ell$, and $z$ using simulation. For each fixed $z$, we derived $\delta$ in Theorem~\ref{thr:improvedChenUpper} accordingly. Since the number of tests with Cheraghchi's scheme and Ahlswede et al.'s scheme is asymptotic while the number of tests with other works is exact, we consider only the other works, which are our proposed theorems, Chen et al.'s scheme, and Chen and Fu's scheme.

Since the number of test with Chen and Fu's scheme is equal to the one with Chen et al.'s scheme, we only consider Chen et al.'s scheme here. Similarly, since the numbers of tests with the four proposed theorems (Theorems~\ref{thr:improvedChenUpper},~\ref{thr:improved0},~\ref{thr:mainImprovement},~\ref{thr:mainImprovement2}) are identical, we only consider the number of tests in Theorem~\ref{thr:improvedChenUpper}. The two schemes are visualized in Figures~\ref{fig:3Schemes}--\ref{fig:2Schemes}. The red and green lines represent for Theorem~\ref{thr:improvedChenUpper} and Chen et al.'s scheme, respectively.

Since Chan et al.~\cite{chan2013stochastic} and Reisizadeh et al.~\cite{reisizadeh2018sub} used a model for the test outcome when the number of defectives in a test fell between $\ell$ and $u$, we do not show the number of tests for their work here. The numbers of tests for Theorem~\ref{thr:improvedChenUpper} and Chen et al.'s scheme are plotted in the figures as $\log_{10}{t}$ versus $\log_{10}{n}$ for various settings of $n, d, u, \ell$, and $z$, where $t$ is the number of tests.

Parameter $z$ was set to $\{3, 11, 101 \}$ corresponding to error tolerance $e = \{1, 5, 50 \}$. The number of items $n$ and the maximum number of defectives $d$ were respectively set to $\{ 10^6 = 1 \mathrm{M}, 10^8 = 10\mathrm{M}, 10^9 = 1\mathrm{B}, 10^{10} = 10\mathrm{B}, 10^{11} = 100\mathrm{B} \}$ and $\{20, 100, 1000 \}$. Finally, upper threshold $u$ and lower threshold $\ell$ were respectively set to $0.2 d$ and $0.5u = 0.1d$.

As shown in Fig.~\ref{fig:3Schemes} for $d = 20$ and Fig.~\ref{fig:2Schemes} for $d = 100$ and $d = 1000$, the number of tests with Theorem~\ref{thr:improvedChenUpper} was the smallest for all settings compared to Chen et al.'s scheme. More importantly, the number was smaller than the number of items (except for $n = 10^6$) while those with the other schemes were mostly larger than the number of items.

When $d = 20$ (Fig.~\ref{fig:3Schemes}), for a small $z$, the number of tests with Chen et al.'s scheme was relatively close to ours. However, as $z$ increased, the number of tests with Chen et al.'s scheme quickly diverged from that with Theorem~\ref{thr:improvedChenUpper}.

\begin{figure}[t]
\centering

\includegraphics[width=1.0\linewidth]{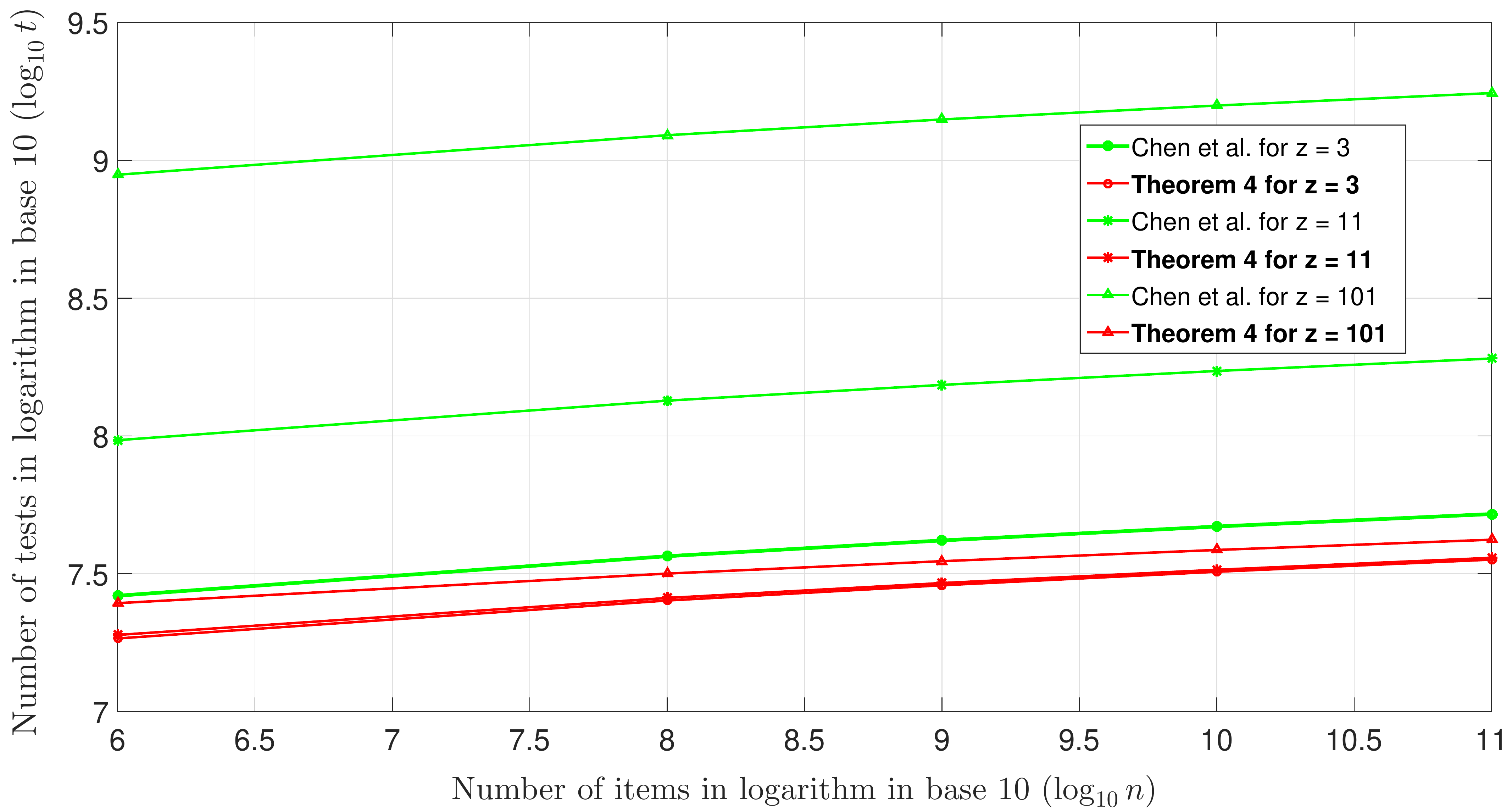}

\caption{Upper bounds on the number of tests versus number of items in logarithm in base 10 for $d = 20$, $z = \{3, 11, 101 \}$, and $n = \{ 10^6 = 1 \mathrm{M}, 10^8 = 10\mathrm{M}, 10^9 = 1\mathrm{B}, 10^{10} = 10\mathrm{B}, 10^{11} = 100\mathrm{B} \}$ for Chen et al.'s scheme and Theorem~\ref{thr:improvedChenUpper}.}
\label{fig:3Schemes}
\end{figure}

\begin{figure}[t]
\begin{subfigure}{0.5\textwidth}
\centering
  \includegraphics[width=1.0\linewidth]{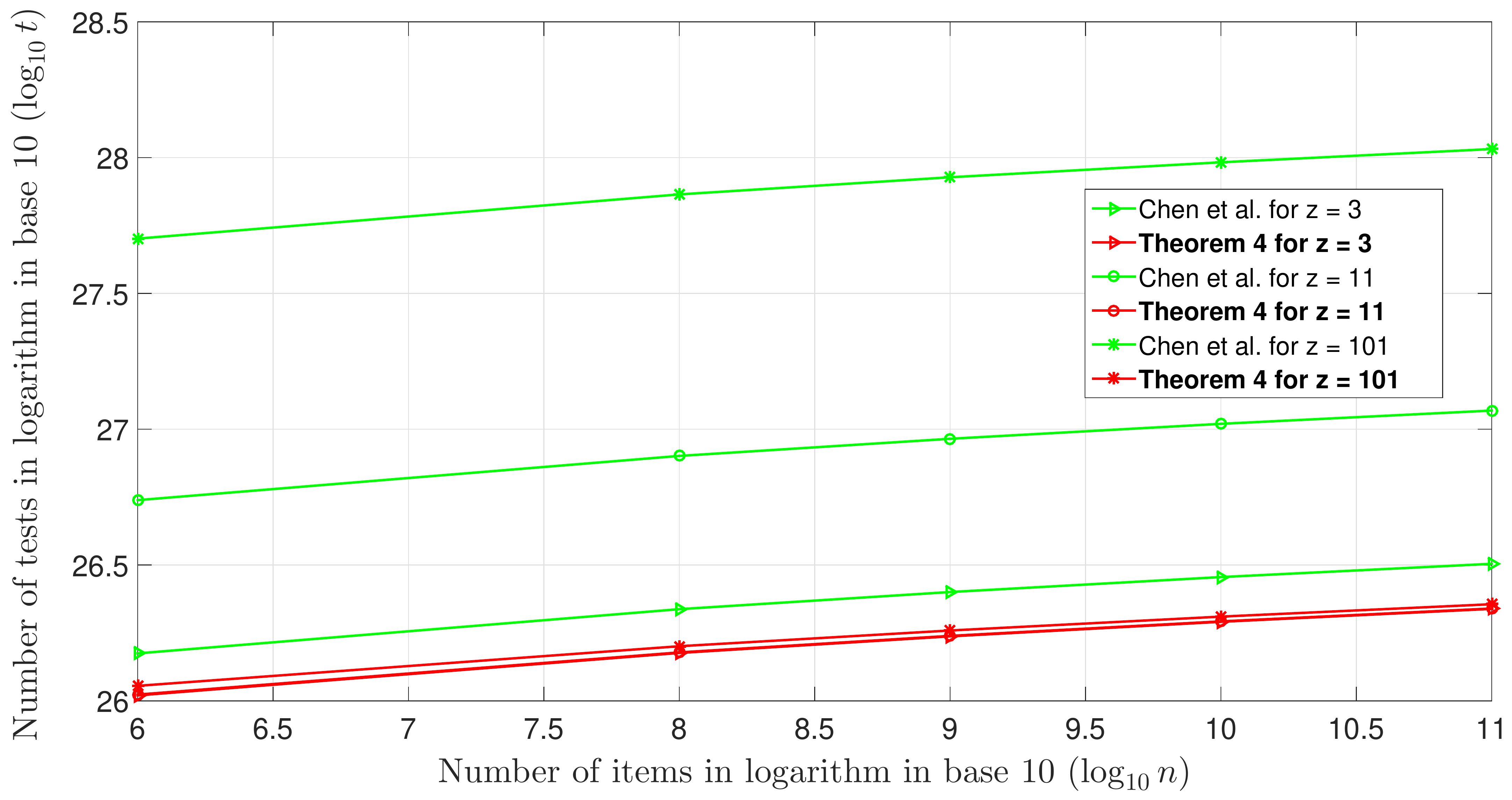}
  \caption{$d = 100$.}
  \label{fig:d_100}
\end{subfigure}
\begin{subfigure}{0.5\textwidth}
  \centering
  \includegraphics[width=1.0\linewidth]{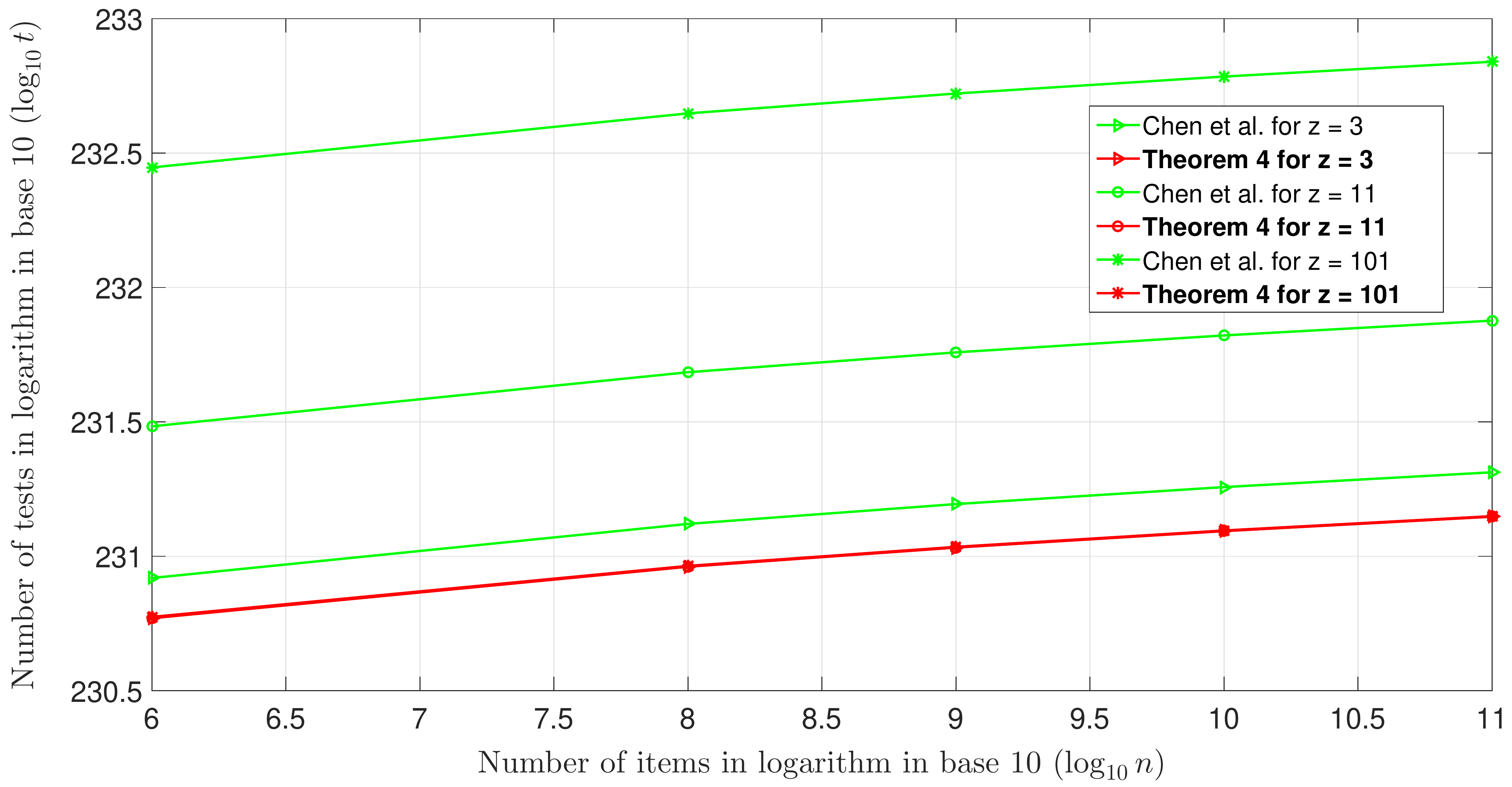}
  \caption{$d = 1000$.}
  \label{fig:d_1000}
\end{subfigure}

\caption{Upper bounds on the number of tests versus number of items in logarithm in base 10 for $d = \{100, 1000\}$, $z = \{3, 11, 101 \}$, and $n = \{ 10^6 = 1 \mathrm{M}, 10^8 = 10\mathrm{M}, 10^9 = 1\mathrm{B}, 10^{10} = 10\mathrm{B}, 10^{11} = 100\mathrm{B} \}$ for Chen et al.'s scheme and Theorem~\ref{thr:improvedChenUpper}.}
\label{fig:2Schemes}
\end{figure}

In summary, the results of simulation match those of our analysis in Section~\ref{main:sub:improvedDisjunct:first}: the upper bound on the number of tests in Theorem~\ref{thr:improvedChenUpper} is always smaller than the one in Theorem~\ref{thr:ChenUpper} for any positive $z$.

\section{Conclusion}
\label{sec:cls}

In this paper, we have presented a novel construction scheme for disjunct matrices that is better than the construction proposed by Chen et al.~\cite{chen2008upper}. For threshold group testing, Cheraghchi gave a hint that the number of tests can be asymptotically to $O(d^{2 + g} \log(n/d) \cdot c_u)$, which is essentially optimal, where $c_u = (8u)^u$. Therefore, it is an interesting question that whether we can reduce the magnitude of the constant $c_u$ and have a decoding algorithm associated with that number of tests.

We next presented a more accurate theorem for Chen and Fu's scheme~\cite{chen2009nonadaptive}, three proposed algorithms on improving non-adaptive encoding and decoding algorithms for threshold group testing as well as simulation for verifying our arguments throughout this work.

\appendix

We use the full expression for~\eqref{eq:ChenUpper} instead of removing $(d - u) \binom{n - u}{g + 1} \binom{d}{g} \binom{d}{u}$ as done by Chen and Fu~\cite{chen2009nonadaptive}. Their inaccurate analysis in the complexity of Step~\ref{alg:getDefectiveSet} led to \textit{inaccurate} decoding complexity in Algorithm~\ref{alg:decodingThreshold}. They presumed that $(d - u) \binom{n - u}{g + 1} \binom{d}{g} \binom{d}{u}$ can be reduced to $O(n^{g + 1})$, and therefore is smaller than $\binom{n}{u} = O(n^u)$.

We first analyze the complexity of Step~\ref{alg:getDefectiveSet}. Let $\alpha$ be the cardinality of $S_i$. We always have $u \leq |S_i| \leq d - 1$ for $i < m$. The time costs of finding all possible subsets $A_i$ and $B_i$ are $\binom{n - \alpha}{g + 1}$ and $\binom{\alpha}{g}$, respectively. One can verify whether ``$\mathbb{H}$ is $u$-complete with respect to $(S_i \cup A_i) \setminus B_i$'' if $t_0^{\cM}(Z) \leq e$ for every $u$-subset $Z \subseteq V$. The complexity of the verification is $\binom{\alpha + 1}{u} \times u \times t(n, d - \ell, u; z]$. Chen and Fu claimed that this cost is $\binom{\alpha + 1}{u} \leq \binom{d}{u}$, which is simply equivalent to the complexity of counting all possibilities of $u$-subsets in $(S_i \cup A_i) \setminus B_i$. This claim is \textit{inaccurate}. Since Step~\ref{alg:getDefectiveSet} is repeated up to $d - u$ times, the complexity of executing this step is
\begin{align}
&(d - u) \binom{n - \alpha}{g + 1} \binom{\alpha}{g} \binom{\alpha + 1}{u} u \times t(n, d - \ell, u; z] \nonumber \\
&= O \left( u(d-u) \binom{n - u}{g + 1} \binom{d - 1}{g} \binom{d}{u} t(n, d - \ell, u; z] \right). \nonumber
\end{align}

We next prove that the quantity $(d-u) \binom{n - u}{g + 1} \binom{d - 1}{g} \binom{d}{u}$ in~\eqref{eq:ChenUpper} should not be removed because it is not always smaller than $\binom{n}{u}$. Let us consider the case in which $u \geq 2$, $d = 2u$, and $u = g + 1$, i.e., $\ell = 0$. We have:
\begin{align}
&(d-u) \binom{n - u}{g + 1} \binom{d - 1}{g} \binom{d}{u} \nonumber \\
&= u \binom{n - u}{u} \binom{2u - 1}{u - 1} \binom{2u}{u} \nonumber \\
&= u \cdot \frac{(n - u)(n - u - 1) \ldots (n - u - (u - 1))}{u!} \cdot \frac{u}{2u(2u - u + 2)} \binom{2u}{u} \cdot \binom{2u}{u} \nonumber \\
&> \frac{(n - 2u + 1)^u}{u!} \cdot \frac{u}{2(u + 2)} \binom{2u}{u}^2 \nonumber \\
&> \frac{(n - 2u + 1)^u}{u!} \cdot \frac{u}{2(u + 2)} \left( \frac{1.08444}{2 \mathrm{e}^{1/(8u)} \sqrt{u}} \cdot 2^{2u} \right)^2, \label{eqn:approxBinom} \\
&> \frac{(n - 2u + 1)^u}{u!} \cdot \frac{1}{2(u + 2)} \cdot \left( \frac{1.08444}{2 \mathrm{e}^{1/(8 \times 2)}} \right)^2 \cdot 16^u \nonumber \\
&> \frac{(n - 2u + 1)^u}{u!} \cdot \frac{1}{7(u + 2)} \cdot 16^u, \nonumber %\label{eq:approxDenominator}
\end{align}
and
\begin{equation}
\binom{n}{u} = \frac{n (n - 1) \ldots (n - (u - 1))}{u!} < \frac{n^u}{u!}, \nonumber
\end{equation}
where~\eqref{eqn:approxBinom} is attained by using the inequality $\binom{mu}{u} > 1.08444 \mathrm{e}^{-1/(8u)} u^{-1/2} \frac{m^{m(u - 1) + 1}}{(m - 1)^{(m - 1)(u - 1)}}$ for integers $m > 1$ and $u \geq 2$ (Corollary 2.9 in~\cite{stanica2001good}). Consider the following inequality:
\begin{align}
&& \frac{(n - 2u + 1)^u}{u!} \cdot \frac{1}{7(u + 2)} \cdot 16^u &\geq \frac{n^u}{u!}  \nonumber \\
&\Longleftrightarrow& 1 - \frac{1}{16} \cdot \left( 7(u + 2) \right)^{1/u} &\geq \frac{2u - 1}{n}. \label{eq:computeN}
\end{align}

Since $\left( 7(u + 2) \right)^{1/u}$ is a decreasing function of $u$ and $u \geq 2$, for~\eqref{eq:computeN} to hold, it suffices that
\begin{align}
&& 1 - \frac{1}{16} \cdot \left( 7(u + 2) \right)^{1/u} \geq 1 - \frac{\sqrt{28}}{16} &\geq \frac{2u - 1}{n} \nonumber \\
&\Longleftrightarrow& n &\geq \frac{8 (2u - 1)}{8 - \sqrt{7}}. \nonumber
\end{align}

Therefore, when $d = 2u, u = g + 1 \geq 2$, and $n \geq \frac{8 (2u - 1)}{8 - \sqrt{7}}$, we always have the following inequality
\begin{align}
(d-u) \binom{n - u}{g + 1} \binom{d - 1}{g} \binom{d}{u} &> \frac{(n - 2u + 1)^u}{u!} \cdot \frac{1}{7(u + 2)} \cdot 16^u \nonumber \\
&\geq \frac{n^u}{u!} > \binom{n}{u}. \nonumber
\end{align}

In summary, the complexity in~\eqref{eq:ChenUpper} is inaccurate.

\section*{Acknowledgment}
\label{sec:ack}

The authors thank Jonathan Scarlett at the National University of Singapore for his constructive and insightful comments on an early version of this paper. We also thank Roghayyeh Haghvirdinezhad for her comments on Fig.~\ref{fig:Overview}. The authors would like to thank the anonymous reviewers for their invaluable comments on an earlier draft of this work.

\bibliographystyle{ieeetr}
\balance
\bibliography{bibli}

\end{document}